\documentclass[12pt]{article}

\usepackage{amsmath,amssymb,amsthm}
\usepackage[textwidth=16cm,textheight=22.6cm]{geometry}
\usepackage[round]{natbib}
\usepackage[colorlinks,linkcolor=blue,citecolor=blue]{hyperref}
\newtheorem{theorem}{Theorem}
\newtheorem{lemma}{Lemma}
\newtheorem{corollary}{Corollary}

\title{Approximation algorithms for identical parallel machine scheduling
with assignment costs}

\author{ Wei Li\\
\normalsize School of Mathematics and Statistics, Yunnan University\\
\normalsize \texttt{wlmath@126.com}}
\date{}

\begin{document}
\maketitle

\begin{abstract}
    We study non-preemptive scheduling on identical parallel machines in
    which assigning job $j$ to machine $M_i$ incurs the cost
    $c_{ij}=a_iw_j+b_j$.
    The processing times $p_j$ and the cost weights $w_j$ are independent,
    and the number of machines is part of the input.
        For every fixed $\varepsilon\in(0,1]$ and every deadline $T$ for which
    a schedule with makespan at most $T$ exists, we give a polynomial-time
    algorithm that computes a schedule with makespan at most
    $(1+\varepsilon)T$ whose total assignment cost does not exceed the
    minimum under the deadline $T$.
    Using this algorithm as a subroutine, we obtain a polynomial-time
    approximation scheme for minimizing the makespan under a total cost
    budget whose schedules always satisfy the budget; this answers an open
    question for costs $c_{ij}=l_ic_j$ with $c_j$ independent of $p_j$.
    For the lexicographic problem, in which the total assignment cost is
    minimized among all schedules of minimum makespan, we obtain in
    polynomial time a schedule with makespan at most $1+\varepsilon$ times
    the minimum and total assignment cost at most that of a
    lexicographically optimal schedule; this improves the makespan
    guarantees of earlier algorithms for additive and multiplicative
    costs.
    For costs $c_{ij}=s_i+a_iw_j$, where $s_i$ is charged for each job
    processed on $M_i$, the same guarantees are obtained in
    quasi-polynomial time.

    \vspace{2mm}\noindent\textbf{Keywords:} Scheduling; Identical
    parallel machines; Machine assignment costs; Approximation schemes;
    Bicriteria approximation; Lexicographic optimization
\end{abstract}

\section{Introduction}
\label{sec:introduction}

In parallel machine scheduling with assignment costs, the machine
selected for a job determines both the load placed on that machine and
the cost of processing the job.
Scheduling decisions must therefore take time and cost into account
simultaneously.
The generalized assignment problem (GAP) provides a general model for
this trade-off \citep{Ross1975,Shmoys1993}.
An instance consists of $m$ parallel machines and $n$ independent
jobs; processing job $j$ on machine $M_i$ requires $p_{ij}$ time units
and incurs a cost $c_{ij}$, and every job must be processed
non-preemptively by exactly one machine.
Given a makespan bound $T$ and a cost bound $\widehat U$, the GAP asks
whether there exists a schedule with makespan at most $T$ and total
cost at most $\widehat U$.
When all costs are zero, the GAP reduces to the unrelated parallel
machine scheduling problem $R\Vert C_{\max}$ and inherits its
inapproximability \citep{Lenstra1990}.
Research has therefore concentrated on bicriteria approximation and on
special cases with additional structure.

Special cases with machine-independent processing times arise in green
manufacturing, where the machines have the same functionality and
speed but differ in energy consumption \citep{Li2016}.
In this setting, $p_{ij}=p_j$ and $c_{ij}=l_ip_j$, where $l_i$ denotes
the processing cost per unit time on machine $M_i$; the cost incurred
on a machine is thus proportional to its load.
Two optimization problems have been studied for this model: the
minimization of the total assignment cost subject to a common deadline
$T$ \citep{Jiang2023,Li2025}, and the minimization of the makespan
subject to a cost budget $\widehat U$ \citep{Li2016,Jiang2023,Li2024}.

\citet{Leung2012} proposed a lexicographic treatment of the two
criteria, motivated by the different objectives of customers and
service providers.
When the makespan is the primary objective, it is minimized first, and
the total assignment cost is then minimized over all schedules that
attain the minimum makespan.
In this line of work, the job cost parameters are specified
independently of the processing times.
Additive costs $c_{ij}=a_i+c_j$ \citep{Leung2012,Lee2014,Kononov2019}
and multiplicative costs $c_{ij}=a_ic_j$ \citep{Leung2012,Lee2014}
have received particular attention, where $a_i$ is a machine-dependent
parameter and $c_j$ is a job-dependent parameter.
\citet{Leung2012} also considered $c_{ij}=\max\{a_i,c_j\}$,
$c_{ij}=\min\{a_i,c_j\}$, and a setting with two classes of machines.

The cost functions $c_{ij}=l_ip_j$, $c_{ij}=a_i+c_j$, and
$c_{ij}=a_ic_j$ are special cases of
\begin{equation}
    c_{ij}=a_iw_j+b_j,
    \label{eq:intro-cost}
\end{equation}
where $a_i\ge0$ is a machine cost coefficient, $w_j\ge0$ is a job cost
weight, and $b_j\ge0$ is a job-specific cost that is independent of
the selected machine.
This paper studies \eqref{eq:intro-cost} with processing times $p_j$
and weights $w_j$ specified independently.
In this case, the load of a machine no longer determines the cost of
the jobs assigned to it, and an assignment must account for processing
times and weights at the same time.
Because $\sum_{j}b_j$ takes the same value for every schedule, the
terms $b_j$ do not affect the set of optimal schedules; they are
retained so that additive costs are obtained with $w_j=1$ and
$b_j=c_j$.
We also study the cost function $c_{ij}=s_i+a_iw_j$, in which a base
cost $s_i\ge0$ is charged for each job processed on $M_i$.

\subsection{Problem definition}
\label{subsec:problem}

We are given a set $\mathcal M=\{M_1,\ldots,M_m\}$ of identical
parallel machines and a set $\mathcal J=\{1,\ldots,n\}$ of independent
jobs.
Machine $M_i$ has a cost coefficient $a_i\ge0$.
Job $j$ has a processing time $p_j>0$, which is the same on every
machine, a cost weight $w_j\ge0$, and a machine-independent cost
$b_j\ge0$; assigning job $j$ to machine $M_i$ incurs the cost
$c_{ij}=a_iw_j+b_j$.
Each job is processed non-preemptively on exactly one machine, each
machine processes at most one job at a time, and machines may remain
idle.
The makespan $C_{\max}$ is the maximum completion time over all jobs,
and the total machine assignment cost ($\operatorname{TMC}$) is the sum
of the assignment costs of all jobs.

Using the three-field notation, we consider three problems.
For a common deadline $T>0$, problem
$P|C_{\max}\le T|\operatorname{TMC}$ asks for a schedule that
minimizes the total assignment cost among all schedules with makespan
at most $T$.
For a cost budget $\widehat U\ge0$, problem
$P|\operatorname{TMC}\le\widehat U|C_{\max}$ asks for a schedule that
minimizes the makespan among all schedules with total assignment cost
at most $\widehat U$.
Problem $\operatorname{LEX}(C_{\max},\operatorname{TMC})$ asks for a
schedule that minimizes the total assignment cost among all schedules
with minimum makespan.

Approximation guarantees are stated as pairs $(\alpha,\beta)$ with
$\alpha,\beta\ge1$.
For $P|C_{\max}\le T|\operatorname{TMC}$, an
$(\alpha,\beta)$-approximation algorithm returns, whenever a schedule
with makespan at most $T$ exists, a schedule with makespan at most
$\alpha T$ and total assignment cost at most $\beta$ times the minimum
total assignment cost among schedules with makespan at most $T$.
For $P|\operatorname{TMC}\le\widehat U|C_{\max}$, it returns, whenever
a schedule with total assignment cost at most $\widehat U$ exists, a
schedule with total assignment cost at most $\alpha\widehat U$ and
makespan at most $\beta$ times the minimum makespan among such
schedules.
In both cases, $\alpha$ measures the violation of the constraint and
$\beta$ measures the objective value relative to the optimum under the
original constraint; in particular, $\alpha=1$ means that the
constraint is satisfied. When $\alpha>1$, the benchmark for $\beta$ remains the optimum under the original constraint, so $\beta=1$ does not imply optimality among
the schedules that satisfy the relaxed constraint.
For $\operatorname{LEX}(C_{\max},\operatorname{TMC})$, an
$(\alpha,\beta)$-approximation algorithm returns a schedule whose
makespan is at most $\alpha$ times the minimum makespan and whose total
assignment cost is at most $\beta$ times the minimum total assignment
cost among schedules with minimum makespan.
The same notation is used for the corresponding problems on unrelated
machines; for example, $R|C_{\max}\le T|\operatorname{TMC}$ denotes
the cost-minimization version of the GAP with a common capacity $T$.
Several results in the literature are stated for a given cost bound
$\widehat U$: if a schedule with makespan at most $T$ and total cost at
most $\widehat U$ exists, a schedule with makespan at most $\alpha T$
and total cost at most $\beta\widehat U$ is returned.
For rational data, such a guarantee is equivalent to an
$(\alpha,\beta)$-approximation for $R|C_{\max}\le T|\operatorname{TMC}$,
by a binary search over $\widehat U$.
Formal definitions are given in Section~\ref{sec:preliminaries}.

\subsection{Related work}
\label{subsec:related}

\paragraph{General assignment costs.}
For $R|C_{\max}\le T|\operatorname{TMC}$, \citet{Lin1992} obtained a
polynomial-time $(2+1/\varepsilon,1+\varepsilon)$-approximation
algorithm for every $\varepsilon>0$, and \citet{Shmoys1993}
subsequently improved the approximation ratio to $(2,1)$.
Their rounding technique also applies to machine-dependent capacities
$T_i$, in which case the load of each machine $M_i$ is at most $2T_i$.
Without costs, \citet{Potts1985} and \citet{Lenstra1990} gave
$2$-approximation algorithms for $R\Vert C_{\max}$ with a fixed and an
arbitrary number of machines, respectively, and \citet{Shchepin2005}
improved the ratio to $2-1/m$.
\citet{Lenstra1990} also proved that, unless $\mathrm{P}=\mathrm{NP}$,
no polynomial-time algorithm for $R\Vert C_{\max}$ achieves a ratio
less than $3/2$, even for restricted assignment instances with
$p_{ij}\in\{p_j,\infty\}$.

Because every schedule satisfies the cost guarantee when all costs are
zero, this result implies that, unless $\mathrm{P}=\mathrm{NP}$, no
polynomial-time $(\alpha,\beta)$-approximation algorithm with
$\alpha<3/2$ exists for $R|C_{\max}\le T|\operatorname{TMC}$, for any
$\beta\ge1$.  For a fixed number of machines, \citet{Jansen2001} and
\citet{Angel2001} developed algorithms that, for every
$\varepsilon>0$, achieve $(1+\varepsilon,1+\varepsilon)$- and
$(1+\varepsilon,1)$-approximations, respectively, in time polynomial
in $n$ and $1/\varepsilon$.
For instances in which the machines fall into a constant number of
types, with machines of the same type having identical processing
times and identical costs, \citet{Nguyen2020} obtained polynomial-time
$(1+\varepsilon,1+\varepsilon)$-approximation algorithms for
$R|\operatorname{TMC}\le\widehat U|C_{\max}$.
They also raised the question of whether a polynomial-time
$(1,1+\varepsilon)$-approximation algorithm, which respects the
budget, exists for this problem.

\paragraph{Identical machines without costs.}
\citet{Hochbaum1987} developed the first polynomial-time approximation
scheme (PTAS) for $P\Vert C_{\max}$ by combining dual approximation,
rounding of processing times, and dynamic programming over
machine configurations.
\citet{Alon1998} extended this approach to a broad class of objective
functions of the machine loads, and efficient polynomial-time
approximation schemes (EPTASs) were subsequently developed by
\citet{Jansen2010} and \citet{Jansen2020}.
Since $P\Vert C_{\max}$ is strongly NP-hard, it admits no fully
polynomial-time approximation scheme (FPTAS) unless
$\mathrm{P}=\mathrm{NP}$ \citep{Garey1978}.
The algorithms in this paper build on the dual approximation framework
of \citet{Hochbaum1987}.
The main difficulty is to ensure that the total assignment cost does
not exceed the optimum under the original deadline while the
processing times are rounded.

\paragraph{Costs proportional to processing times.}
For $c_{ij}=l_ip_j$ with $l_i>0$, let $\gamma=\max_il_i/\min_il_i$.
For $P|C_{\max}\le T|\operatorname{TMC}$, \citet{Jiang2023} gave a
polynomial-time $(1,(2+\gamma)/3)$-approximation algorithm, and
\citet{Li2025} improved the guarantee to $(1,1+(\gamma-1)\eta)$ for any
fixed $\eta\in(0,1/3]$.
\citet{Li2025} further showed that, when $\gamma$ is unbounded, no
polynomial-time $(1,\beta)$-approximation algorithm with constant
$\beta$ exists unless $\mathrm{P}=\mathrm{NP}$, and obtained a
linear-time $(1,1+\varepsilon)$-approximation algorithm for the case in
which the maximum unit processing cost is a fixed constant.
For $P|\operatorname{TMC}\le\widehat U|C_{\max}$, \citet{Li2016} gave
an $O(n\log n)$-time $(1,2)$-approximation algorithm for the
non-preemptive problem and a linear-time exact algorithm for the
preemptive problem.
\citet{Jiang2023} improved the guarantee for the non-preemptive problem
to $(1,(\sqrt{33}+1)/4)$, where $(\sqrt{33}+1)/4\approx1.686$.
\citet{Li2024} then obtained an $O(n\log n/\varepsilon)$-time
$(1,4/3+\varepsilon)$-approximation algorithm and the first PTAS for
this problem, that is, a polynomial-time
$(1,1+\varepsilon)$-approximation algorithm for each fixed
$\varepsilon>0$; for a fixed number of machines, they also developed
an FPTAS.
They also proposed the development of PTASs for two extensions:
$c_{ij}=s_i+l_ip_j$, where $s_i\ge0$ is a setup cost charged for each
job processed on $M_i$, and $c_{ij}=l_ic_j$, where $c_j\ge0$ is
specified independently of $p_j$.
\citet{Gong2026} studied a related model in which jobs may be rejected
at a penalty and the machines have different energy consumption rates.

\paragraph{Lexicographic optimization.}
For $F\in\{C_{\max},\sum_jC_j\}$, \citet{Leung2012} studied both
$\operatorname{LEX}(F,\operatorname{TMC})$ and
$\operatorname{LEX}(\operatorname{TMC},F)$.
In the non-preemptive setting, both problems with $F=\sum_jC_j$ are
solvable in $O(mn^2+n^2\log n)$ time on unrelated machines, whereas
both problems with $F=C_{\max}$ are strongly NP-hard even on identical
machines.
For $\operatorname{LEX}(C_{\max},\operatorname{TMC})$ on identical
machines, \citet{Lee2014} proposed heuristics that are
$(r+1,1)$-approximation algorithms for additive and multiplicative
costs, where $r\in(1,2)$ is the approximation ratio of the algorithm
used for $P\Vert C_{\max}$.
For additive costs, \citet{Kononov2019} improved the guarantee to
$(3r/2,1)$.
\citet{Kononov2019} also studied
$P|C_{\max}\le\Delta C_{\max}^*|\operatorname{TMC}$, where
$C_{\max}^*$ is the optimal makespan of $P\Vert C_{\max}$ and
$\Delta\ge1$ is a fixed constant.
They proved that this problem is strongly NP-hard for both additive
and multiplicative costs and that, unless $\mathrm{P}=\mathrm{NP}$, it
admits no pseudo-polynomial-time $(1,\beta)$-approximation algorithm
for any finite $\beta$.
For additive costs with $0<a_1\le\cdots\le a_m$ and
$\Delta\ge1+r-1/m$, they gave a polynomial-time
$(1,\beta)$-approximation algorithm in which $\beta$ depends on
$\Delta$, $r$, and $a_m/a_1$, together with an $O(mn^{2k})$-time
dynamic programming algorithm for instances with a fixed number $k$ of
distinct processing times.
\citet{Lee2014} additionally studied
$\operatorname{LEX}(\sum_jC_j,\operatorname{MMC})$, where
$\operatorname{MMC}$ denotes the maximum total assignment cost
incurred on a single machine.

\paragraph{Relation to the present work.}
Several of the results above already yield guarantees of the type
considered here, either for particular cost functions or for a fixed
number of machines.
For fixed $m$, the algorithm of \citet{Angel2001} provides a
$(1+\varepsilon,1)$-approximation for arbitrary costs $c_{ij}$;
combined with the search procedures in Sections~\ref{sec:ptas-budget}
and~\ref{sec:lex-approximation}, it yields the guarantees of this
paper for all three problems.
The results of this paper are therefore of interest when $m$ is part
of the input.
The algorithms of \citet{Nguyen2020} require identical costs within
each of a constant number of machine types, a condition that fails
when the coefficients $a_i$ take more than a constant number of
distinct values, and they provide $(1+\varepsilon,1+\varepsilon)$-
rather than $(1,1+\varepsilon)$-approximations.
For $c_{ij}=l_ip_j$ and integral data, the
$(1,1+\varepsilon)$-approximation algorithm of \citet{Li2024}, combined
with a binary search over the budget, yields a
$(1+\varepsilon,1)$-approximation algorithm for
$P|C_{\max}\le T|\operatorname{TMC}$ and hence for
$\operatorname{LEX}(C_{\max},\operatorname{TMC})$.
The novelty of the present work therefore lies in allowing the weights
$w_j$ to be independent of the processing times $p_j$, which covers
the additive and multiplicative costs of \citet{Leung2012} and the
extension $c_{ij}=l_ic_j$ proposed by \citet{Li2024}.

The structure of the cost function cannot be dispensed with.
For arbitrary costs $c_{ij}$, even with machine-independent processing
times, the restricted assignment problem is obtained by setting
$c_{ij}=0$ for permitted assignments, $c_{ij}=1$ otherwise, and
$\widehat U=0$.
By the result of \citet{Lenstra1990}, there is then no polynomial-time
$(1,\beta)$-approximation algorithm with $\beta<3/2$ for
$P|\operatorname{TMC}\le\widehat U|C_{\max}$ and no polynomial-time
$(\alpha,1)$-approximation algorithm with $\alpha<3/2$ for
$P|C_{\max}\le T|\operatorname{TMC}$, unless $\mathrm{P}=\mathrm{NP}$.
This does not conflict with the results of this paper.
Under \eqref{eq:intro-cost}, $c_{ij}-c_{i'j}=(a_i-a_{i'})w_j$ for every
job $j$, so all jobs rank the machines in the same order of $a_i$.
Such costs cannot express machine restrictions in which two jobs
prefer different machines, as the instances in the above reduction
require.
The inapproximability results of \citet{Kononov2019} and \citet{Li2025}
concern $(1,\beta)$-approximations for deadline-constrained problems
and therefore do not apply to the $(1+\varepsilon,1)$-approximations
developed here.

\subsection{Our contributions}
\label{subsec:contributions}

The main results concern the cost function \eqref{eq:intro-cost} with
processing times and weights specified independently and with the
number of machines part of the input.

For each fixed $0<\varepsilon\le1$, we develop a polynomial-time
$(1+\varepsilon,1)$-approximation algorithm for
$P|C_{\max}\le T|\operatorname{TMC}$.
Whenever a schedule meeting the deadline $T$ exists, the algorithm
returns a schedule with makespan at most $(1+\varepsilon)T$ whose total
assignment cost does not exceed the minimum under the original
deadline.
Following the dual approximation framework, the algorithm rounds the
processing times of large jobs and selects machine configurations by
dynamic programming, while the individual jobs are chosen by two
ordering rules.
Within each rounded processing time class, large jobs are assigned to
the machines in non-increasing order of weight.
Small jobs are assigned in non-increasing order of $w_j/p_j$ according
to cumulative, rather than machine-wise, processing time bounds.
Under these rules, the machine index and the cumulative numbers of
large jobs assigned from each class determine exactly which small jobs
have been assigned, so that the dynamic program requires no state
component for small jobs.
An exchange argument based on prefix sums of weights shows that every
schedule meeting the deadline $T$ can be transformed into a schedule
generated by these rules without increasing its total assignment cost,
at the expense of exceeding $T$ by a factor of at most
$1+\varepsilon$.

Using this algorithm as a subroutine, we obtain, for each fixed
$0<\varepsilon\le1$, a polynomial-time
$(1,1+\varepsilon)$-approximation algorithm for
$P|\operatorname{TMC}\le\widehat U|C_{\max}$ and a polynomial-time
$(1+\varepsilon,1)$-approximation algorithm for
$\operatorname{LEX}(C_{\max},\operatorname{TMC})$.
The former satisfies the original budget and is therefore a PTAS; in
particular, it provides the PTAS for $c_{ij}=l_ic_j$ sought by
\citet{Li2024}.
For multiplicative costs, the latter improves the
$(r+1,1)$-approximation of \citet{Lee2014} to a
$(1+\varepsilon,1)$-approximation; for additive costs, it improves the
$(3r/2,1)$-approximation of \citet{Kononov2019} in the same way.
The algorithms of \citet{Lee2014} and \citet{Kononov2019} have
substantially lower running times, so the improvement concerns the
approximation guarantee rather than computational efficiency.

For $c_{ij}=s_i+a_iw_j$, with no ordering relation imposed between
$s_i$ and $a_i$, the total base cost depends on the number of jobs
assigned to each machine and cannot be controlled through the weights
alone.
We assign every job with processing time at most $\varepsilon T/(2n)$
to a machine of minimum assignment cost and schedule the remaining
jobs, which fall into $O(\varepsilon^{-1}\log(n/\varepsilon))$ rounded
classes, by dynamic programming.
This yields quasi-polynomial-time algorithms with the same guarantees
for all three problems.
In particular, it gives a quasi-polynomial-time approximation scheme
(QPTAS) that satisfies the original budget for $c_{ij}=s_i+l_ip_j$,
which partially answers the second question of \citet{Li2024}; the
existence of a PTAS for this cost function remains open.

The remainder of the paper is organized as follows.
Section~\ref{sec:preliminaries} introduces the notation, defines the
approximation guarantees formally, and presents a lemma for comparing
total assignment costs.
Section~\ref{sec3} presents and analyzes the algorithm for
$P|C_{\max}\le T|\operatorname{TMC}$.
Using this algorithm as a subroutine,
Sections~\ref{sec:ptas-budget} and~\ref{sec:lex-approximation} develop
the algorithms for $P|\operatorname{TMC}\le\widehat U|C_{\max}$ and
$\operatorname{LEX}(C_{\max},\operatorname{TMC})$, respectively.
Section~\ref{sec:affine-qptas} extends the approach to the cost function
$c_{ij}=s_i+a_iw_j$, and Section~\ref{sec:conclusion} concludes with a
discussion of open problems.

\section{Preliminaries}
\label{sec:preliminaries}

\paragraph{Instances and schedules.}
For a positive integer $k$, let $[k]=\{1,\ldots,k\}$.
We use the model of Section~\ref{subsec:problem} and write
$I=(\mathcal M,\mathcal J;p,a,w,b)$ for the machine and job data, where
$p=(p_1,\ldots,p_n)$, $a=(a_1,\ldots,a_m)$, $w=(w_1,\ldots,w_n)$, and
$b=(b_1,\ldots,b_n)$ are the vectors of processing times, machine cost
coefficients, job cost weights, and job-specific costs, respectively.
An instance of $P|C_{\max}\le T|\operatorname{TMC}$ is a pair $(I,T)$
with $T>0$, an instance of $P|\operatorname{TMC}\le\widehat U|C_{\max}$
is a pair $(I,\widehat U)$ with $\widehat U\ge0$, and an instance of
$\operatorname{LEX}(C_{\max},\operatorname{TMC})$ is given by $I$ alone.
For each of these instances, $|I|$ denotes the length of its binary
encoding, including $T$ or $\widehat U$ where applicable.
By reindexing the machines, we assume without loss of generality that
$0\le a_1\le\cdots\le a_m$.

For $A\subseteq\mathcal J$, let $P(A)=\sum_{j\in A}p_j$ and
$W(A)=\sum_{j\in A}w_j$ denote the total processing time and the total
weight of $A$, respectively, with $P(\varnothing)=W(\varnothing)=0$.
A schedule is an $m$-tuple $\sigma=(B_1,\ldots,B_m)$ of pairwise
disjoint, possibly empty, subsets of $\mathcal J$ whose union is
$\mathcal J$, where $B_i$ is the set of jobs assigned to machine $M_i$.
A partial schedule is defined in the same way, except that the union of
its sets may be a proper subset of $\mathcal J$.
Since neither objective depends on the order in which the jobs in
$B_i$ are processed, a schedule is determined by its assignment of
jobs to machines.
Let $\mathcal S$ denote the set of all schedules.
The makespan and the total assignment cost of $\sigma\in\mathcal S$ are
\begin{equation}
    C_{\max}(\sigma)=\max_{i\in[m]}P(B_i)
    \quad\text{and}\quad
    Z(\sigma)=\sum_{i\in[m]}\sum_{j\in B_i}c_{ij}
    =\sum_{i\in[m]}a_iW(B_i)+C_0,
    \label{eq:prelim-cost}
\end{equation}
where $C_0=\sum_{j\in\mathcal J}b_j$ does not depend on the assignment.

\paragraph{Optimal values and basic bounds.}
Throughout the paper, the minimum over an empty set is defined as
$+\infty$.
For $T>0$, let
\[
    \mathcal F(T)=\{\sigma\in\mathcal S:C_{\max}(\sigma)\le T\}
    \quad\text{and}\quad
    Z^*(T)=\min_{\sigma\in\mathcal F(T)}Z(\sigma)
\]
denote the set of schedules that meet the deadline $T$ and the minimum
total assignment cost over this set.
For $\widehat U\ge0$, the optimal makespan of
$P|\operatorname{TMC}\le\widehat U|C_{\max}$ is
\[
    C^*(\widehat U)
    =\min\{C_{\max}(\sigma):\sigma\in\mathcal S,\ Z(\sigma)\le\widehat U\}.
\]
For $\operatorname{LEX}(C_{\max},\operatorname{TMC})$, let
$C^*=\min_{\sigma\in\mathcal S}C_{\max}(\sigma)$ denote the minimum
makespan, and let $Z^*=Z^*(C^*)$ denote the minimum total assignment
cost among schedules with makespan $C^*$.
A schedule $\sigma^*$ with $C_{\max}(\sigma^*)=C^*$ and
$Z(\sigma^*)=Z^*$ is called lexicographically optimal.
Since $\mathcal F(T_1)\subseteq\mathcal F(T_2)$ whenever
$0<T_1\le T_2$, the function $Z^*$ is non-increasing, that is,
\begin{equation}
    Z^*(T_2)\le Z^*(T_1),
    \qquad 0<T_1\le T_2.
    \label{eq:prelim-monotone}
\end{equation}

Let $p_{\max}=\max_{j\in\mathcal J}p_j$ and
$L=\max\{p_{\max},P(\mathcal J)/m\}$.
Every schedule $\sigma\in\mathcal S$ satisfies
\begin{equation}
    L\le C_{\max}(\sigma)\le P(\mathcal J),
    \label{eq:prelim-makespan-bounds}
\end{equation}
because its makespan is at least the longest processing time and at
least the average machine load, and at most the total processing time.
In particular, $\mathcal F(T)\ne\varnothing$ implies $p_{\max}\le T$
and $P(\mathcal J)\le mT$.
Moreover, the definition of $L$ and the inequality
$P(\mathcal J)\le np_{\max}$ give
$1\le P(\mathcal J)/L\le\min\{m,n\}$.

Since $a_1\le a_i$ for every $i\in[m]$ and $w_j\ge0$ for every
$j\in\mathcal J$, every schedule satisfies
$Z(\sigma)\ge a_1W(\mathcal J)+C_0$.
Equality holds for the schedule $(\mathcal J,\varnothing,\ldots,
\varnothing)$, which assigns all jobs to $M_1$ and has makespan
$P(\mathcal J)$.
Consequently, $P|\operatorname{TMC}\le\widehat U|C_{\max}$ has a
schedule satisfying the budget if and only if
$a_1W(\mathcal J)+C_0\le\widehat U$, in which case
$L\le C^*\le C^*(\widehat U)\le P(\mathcal J)$.
The same schedule shows that $Z^*(T)=a_1W(\mathcal J)+C_0$ for every
$T\ge P(\mathcal J)$.

\paragraph{Approximation guarantees.}
We now state formally the guarantees described in
Section~\ref{subsec:problem}.
Let $\alpha,\beta\ge1$.
An algorithm for $P|C_{\max}\le T|\operatorname{TMC}$ is an
$(\alpha,\beta)$-approximation algorithm if, for every instance
$(I,T)$ with $\mathcal F(T)\ne\varnothing$, it returns a schedule
$\widehat\sigma\in\mathcal S$ such that
\[
    C_{\max}(\widehat\sigma)\le\alpha T
    \quad\text{and}\quad
    Z(\widehat\sigma)\le\beta Z^*(T).
\]
An algorithm for $P|\operatorname{TMC}\le\widehat U|C_{\max}$ is an
$(\alpha,\beta)$-approximation algorithm if, for every instance
$(I,\widehat U)$ with $C^*(\widehat U)<+\infty$, it returns a schedule
$\widehat\sigma\in\mathcal S$ such that
\[
    Z(\widehat\sigma)\le\alpha\widehat U
    \quad\text{and}\quad
    C_{\max}(\widehat\sigma)\le\beta C^*(\widehat U).
\]
On an instance that admits no schedule satisfying the original
constraint, such an algorithm may return an arbitrary schedule or
report that no feasible schedule exists.
An algorithm for $\operatorname{LEX}(C_{\max},\operatorname{TMC})$ is
an $(\alpha,\beta)$-approximation algorithm if, for every instance $I$,
it returns a schedule $\widehat\sigma\in\mathcal S$ such that
\[
    C_{\max}(\widehat\sigma)\le\alpha C^*
    \quad\text{and}\quad
    Z(\widehat\sigma)\le\beta Z^*.
\]
These definitions do not restrict the running time, which is stated
separately for each algorithm.
In Section~\ref{sec:affine-qptas}, they are used with the total
assignment cost defined by the cost function considered there.

\paragraph{Approximation schemes.}
Given a minimization problem, an $\alpha$-approximation algorithm
produces, for every feasible instance $I$, a feasible solution
with objective value at most $\alpha$ times the optimal
objective value in time polynomial in the input size $|I|$.
A quasi-polynomial-time approximation scheme (QPTAS) is a
family of algorithms such that, for every fixed
$0<\varepsilon\le1$, it contains an algorithm with
approximation ratio $1+\varepsilon$ and running time
$|I|^{(\log |I|)^{O(1)}}$.
If the running time is polynomial in $|I|$ for each fixed
$\varepsilon$, the scheme is called a polynomial-time
approximation scheme (PTAS).
A PTAS is called an efficient polynomial-time approximation
scheme (EPTAS) if its running time is
$f(1/\varepsilon)|I|^c$, where $c$ is a constant independent
of $\varepsilon$ and $f$ is a computable function.
An EPTAS is called a fully polynomial-time approximation
scheme (FPTAS) if $f(1/\varepsilon)$ is polynomial in
$1/\varepsilon$, so that its running time is polynomial
in both $|I|$ and $1/\varepsilon$.
For $P|\operatorname{TMC}\le\widehat U|C_{\max}$, the budget is a
constraint and the makespan is the objective, so a family of
$(1,1+\varepsilon)$-approximation algorithms, which satisfy the budget,
is an approximation scheme in this sense.
For $P|C_{\max}\le T|\operatorname{TMC}$ and
$\operatorname{LEX}(C_{\max},\operatorname{TMC})$, the guarantees
obtained in this paper allow the makespan to exceed the deadline or
the minimum makespan by a factor of $1+\varepsilon$, and we do not
refer to the corresponding algorithms as approximation schemes.

\paragraph{Comparison of total assignment costs.}
The following lemma gives a sufficient condition for comparing the
total assignment costs of two schedules of the same instance.
It is an equivalent form of Lemma~2.1 in \citet{Li2024}, stated in
terms of prefix sums; a short proof is included for completeness.

\begin{lemma}
    \label{lem:mctcm-prefix-cost}
    Let $a_1\le\cdots\le a_m$, and let $(x_1,\ldots,x_m)$ and
    $(y_1,\ldots,y_m)$ be real vectors such that
    $\sum_{i=1}^{m}x_i=\sum_{i=1}^{m}y_i$ and
    $\sum_{i=1}^{r}x_i\ge\sum_{i=1}^{r}y_i$ for $r=1,\ldots,m-1$.
    Then $\sum_{i=1}^{m}a_ix_i\le\sum_{i=1}^{m}a_iy_i$.
\end{lemma}

\begin{proof}
    Let $D_0=0$ and $D_r=\sum_{i=1}^{r}(x_i-y_i)$ for $r\in[m]$.
    By assumption, $D_m=0$ and $D_r\ge0$ for $r\in[m-1]$.
    Summation by parts gives
    \[
        \sum_{i=1}^{m}a_i(x_i-y_i)
        =\sum_{i=1}^{m}a_i(D_i-D_{i-1})
        =a_mD_m-\sum_{r=1}^{m-1}(a_{r+1}-a_r)D_r
        \le0,
    \]
    since $a_{r+1}-a_r\ge0$ and $D_r\ge0$ for every $r\in[m-1]$.
\end{proof}

\begin{corollary}
    \label{cor:prelim-prefix-weight}
    Let $\sigma=(B_1,\ldots,B_m)$ and $\sigma'=(B_1',\ldots,B_m')$ be
    schedules of the same instance.
    If $\sum_{i=1}^{r}W(B_i)\ge\sum_{i=1}^{r}W(B_i')$ for
    $r=1,\ldots,m-1$, then $Z(\sigma)\le Z(\sigma')$.
\end{corollary}

\begin{proof}
    Both schedules assign every job exactly once, so
    $\sum_{i=1}^{m}W(B_i)=\sum_{i=1}^{m}W(B_i')=W(\mathcal J)$.
    Applying Lemma~\ref{lem:mctcm-prefix-cost} with $x_i=W(B_i)$ and
    $y_i=W(B_i')$ gives
    $\sum_{i=1}^{m}a_iW(B_i)\le\sum_{i=1}^{m}a_iW(B_i')$.
    The claim follows from \eqref{eq:prelim-cost}, because $C_0$ is the
    same for both schedules.
\end{proof}

\section{An approximation algorithm for \texorpdfstring$P|C_{\max}\le T|\operatorname{TMC}$}
\label{sec3}

Throughout this section, the deadline $T>0$ and the accuracy parameter
$0<\varepsilon\le1$ are fixed, and we set
\[
    \delta=\varepsilon/2,
    \qquad
    \tau=\delta T,
    \qquad
    H=(1+\delta)T.
\]
By Section~\ref{sec:preliminaries}, $\mathcal F(T)=\varnothing$ if
$p_{\max}>T$ or $P(\mathcal J)>mT$.
We therefore assume that $p_{\max}\le T$ and $P(\mathcal J)\le mT$, but
we do not assume that $\mathcal F(T)\ne\varnothing$.
Following the dual approximation framework of \citet{Hochbaum1987}, the
algorithm rounds the processing times of large jobs and selects machine
configurations by dynamic programming.
Each sequence of configurations determines a single schedule through
two ordering rules, and the dynamic program selects the cheapest of
these schedules.
Section~\ref{subsec:sec3-rounding} introduces the rounding and the
configurations, and Sections~\ref{subsec:sec3-large}
and~\ref{subsec:sec3-small} present the ordering rules for large and
small jobs.
Section~\ref{subsec:sec3-canonical} compares the resulting schedules
with the schedules in $\mathcal F(T)$, and
Sections~\ref{subsec:sec3-dp} and~\ref{subsec:sec3-algorithm} give the
dynamic program and the analysis of the algorithm.

\subsection{Rounding and configurations}
\label{subsec:sec3-rounding}

A job $j$ is large if $p_j>\tau$ and small otherwise; let
$\mathcal J_L$ and $\mathcal J_S$ denote the sets of large and small
jobs, respectively.
Since $P(\mathcal J)\le mT$ and every large job has processing time
greater than $\delta T$, we have $|\mathcal J_L|<m/\delta$.

Let $\pi_t=(1+\delta)^t\tau$ for $t=0,1,\ldots$, and let $K$ be the
smallest index with $\pi_K\ge T$.
Then $\pi_{K-1}<T\le\pi_K$ and
\[
    K=\left\lceil\frac{\log(1/\delta)}{\log(1+\delta)}\right\rceil
    =O\bigl(\delta^{-1}\log(1/\delta)\bigr),
\]
where the bound follows from $\log(1+\delta)\ge\delta/(1+\delta)$.
Since $\tau<p_j\le T$ for every large job, the classes
\begin{equation}
    \mathcal J_t=\{j\in\mathcal J_L:\pi_{t-1}<p_j\le\pi_t\},
    \qquad t\in[K],
    \label{eq:mctcm-job-classes}
\end{equation}
are pairwise disjoint and cover $\mathcal J_L$; some of them may be
empty.
Let $N_t=|\mathcal J_t|$ and $\boldsymbol N=(N_1,\ldots,N_K)$.
The processing time of each job $j\in\mathcal J_t$ is rounded up to
$\bar p_j=\pi_t$, whereas small jobs keep their processing times.
Since $p_j\le\bar p_j=(1+\delta)\pi_{t-1}<(1+\delta)p_j$ for
$j\in\mathcal J_t$, every $A\subseteq\mathcal J_L$ satisfies
\begin{equation}
    P(A)\le\sum_{j\in A}\bar p_j\le(1+\delta)P(A).
    \label{eq:mctcm-rounding-bound}
\end{equation}

All jobs in $\mathcal J_t$ have the same rounded processing time, so the
rounded load of the large jobs on a machine depends only on the number
of jobs taken from each class.
A configuration is a vector $\boldsymbol q=(q_1,\ldots,q_K)\in
\mathbb Z_{\ge0}^K$, where $q_t$ is the number of jobs from
$\mathcal J_t$ on the machine.
For $\boldsymbol x\in\mathbb Z_{\ge0}^K$, let
$\bar P(\boldsymbol x)=\sum_{t=1}^{K}x_t\pi_t$, and write
$\boldsymbol x\le\boldsymbol y$ if $x_t\le y_t$ for all $t\in[K]$.
With $\kappa=\lfloor(1+\delta)/\delta\rfloor$, we use the set of
configurations
\begin{equation}
    \mathcal Q=\Bigl\{
    \boldsymbol q\in\mathbb Z_{\ge0}^K:
    \boldsymbol q\le\boldsymbol N,\
    \sum_{t=1}^{K}q_t\le\kappa,\
    \bar P(\boldsymbol q)\le H
    \Bigr\}.
    \label{eq:mctcm-configurations}
\end{equation}
Because $\pi_t>\delta T$ for every $t\in[K]$, the condition
$\sum_tq_t\le\kappa$ is implied by $\bar P(\boldsymbol q)\le H$.
It is stated explicitly because $\mathcal Q$ is constructed by
enumerating the $\binom{K+\kappa}{K}$ nonnegative integer vectors with
component sum at most $\kappa$.
Hence $|\mathcal Q|\le\binom{K+\kappa}{K}\le(\kappa+1)^K$, which is a
constant for fixed $\varepsilon$.
The zero vector belongs to $\mathcal Q$.

A sequence $\boldsymbol q_1,\ldots,\boldsymbol q_m\in\mathcal Q$ with
$\sum_{i=1}^{m}\boldsymbol q_i=\boldsymbol N$ is called a configuration
sequence, where $\boldsymbol q_i=(q_{i,1},\ldots,q_{i,K})$ is the
configuration of $M_i$.
Every schedule $\sigma=(B_1,\ldots,B_m)\in\mathcal F(T)$ induces a
configuration sequence.
Indeed, let $q_{i,t}=|B_i\cap\mathcal J_t|$.
By \eqref{eq:mctcm-rounding-bound},
$\bar P(\boldsymbol q_i)\le(1+\delta)P(B_i\cap\mathcal J_L)\le H$, so
$\boldsymbol q_i\in\mathcal Q$ for every $i\in[m]$, and
$\sum_{i=1}^{m}\boldsymbol q_i=\boldsymbol N$ because every large job
is assigned exactly once.

\subsection{Assignment of large jobs}
\label{subsec:sec3-large}

A configuration specifies only the number of jobs taken from each
class.
Since the jobs of a class may have different weights, the individual
jobs are selected by the following rule.
For each nonempty class $\mathcal J_t$, let
$\ell_{t,1},\ldots,\ell_{t,N_t}$ be its jobs in non-increasing order of
weight, with ties broken by increasing job index.
Given configurations $\boldsymbol q_1,\boldsymbol q_2,\ldots\in\mathcal Q$
of the machines $M_1,M_2,\ldots$, let
$\boldsymbol z^0=\boldsymbol0$ and
$\boldsymbol z^i=\sum_{k=1}^{i}\boldsymbol q_k$, so that $z_t^i$ is
the number of jobs from $\mathcal J_t$ assigned to $M_1,\ldots,M_i$.
The machines are considered in the order $M_1,\ldots,M_m$, and machine
$M_i$ receives the next $q_{i,t}$ jobs of each list:
\begin{equation}
    L_i=\bigcup_{t=1}^{K}\{\ell_{t,r}:z_t^{i-1}<r\le z_t^i\},
    \qquad i\in[m].
    \label{eq:mctcm-large-assignment}
\end{equation}
For a configuration sequence, $\boldsymbol z^m=\boldsymbol N$, so the
sets $L_1,\ldots,L_m$ are pairwise disjoint with union $\mathcal J_L$,
and $|L_i\cap\mathcal J_t|=q_{i,t}$ for all $i$ and $t$.
In particular, $P(L_i)\le\sum_{j\in L_i}\bar p_j=\bar P(\boldsymbol q_i)
\le H$, regardless of which jobs are selected.
Moreover, $L_i$ depends only on $\boldsymbol q_1,\ldots,
\boldsymbol q_i$.
The next lemma shows that, among all assignments of large jobs with the
same numbers of jobs from each class on each machine, the rule
\eqref{eq:mctcm-large-assignment} maximizes every prefix sum of the
machine weights and therefore minimizes the total assignment cost of
large jobs.

\begin{lemma}
    \label{lem:mctcm-large-order}
    Let $L_1,\ldots,L_m$ be given by \eqref{eq:mctcm-large-assignment}
    for a configuration sequence, and let $L_1',\ldots,L_m'$ be pairwise
    disjoint, possibly empty, subsets of $\mathcal J_L$ with union
    $\mathcal J_L$ such that $|L_i'\cap\mathcal J_t|=q_{i,t}$ for all
    $i\in[m]$ and $t\in[K]$.
    Then
    \begin{equation}
        \sum_{i=1}^{r}W(L_i)\ge\sum_{i=1}^{r}W(L_i'),
        \qquad r\in[m],
        \label{eq:mctcm-large-order-prefix}
    \end{equation}
    and consequently
    $\sum_{i=1}^{m}\sum_{j\in L_i}c_{ij}
        \le\sum_{i=1}^{m}\sum_{j\in L_i'}c_{ij}$.
\end{lemma}

\begin{proof}
    Fix $r\in[m]$ and $t\in[K]$.
    Both assignments place exactly $z_t^r$ jobs of $\mathcal J_t$ on
    $M_1,\ldots,M_r$, and \eqref{eq:mctcm-large-assignment} places the
    first $z_t^r$ jobs of the list $\ell_{t,1},\ldots,\ell_{t,N_t}$.
    These jobs have the largest total weight among all subsets of
    $\mathcal J_t$ of this size, so
    $\sum_{i=1}^{r}W(L_i\cap\mathcal J_t)
        \ge\sum_{i=1}^{r}W(L_i'\cap\mathcal J_t)$.
    Summing over $t$ gives \eqref{eq:mctcm-large-order-prefix}.
    Since $\sum_{i=1}^{m}W(L_i)=\sum_{i=1}^{m}W(L_i')=W(\mathcal J_L)$,
    Lemma~\ref{lem:mctcm-prefix-cost} with $x_i=W(L_i)$ and
    $y_i=W(L_i')$ yields
    $\sum_{i=1}^{m}a_iW(L_i)\le\sum_{i=1}^{m}a_iW(L_i')$.
    Adding $\sum_{j\in\mathcal J_L}b_j$ to both sides proves the cost
    inequality.
\end{proof}

\subsection{Assignment of small jobs}
\label{subsec:sec3-small}

For configurations $\boldsymbol q_1,\boldsymbol q_2,\ldots\in\mathcal Q$
as in Section~\ref{subsec:sec3-large}, let
$R_i=H-\bar P(\boldsymbol q_i)$.
Then $0\le R_i\le H-P(L_i)$, so $R_i$ does not exceed the capacity that
remains on $M_i$ below $H$ after the large jobs are assigned.

Let $N_S=|\mathcal J_S|$, and let $s_1,\ldots,s_{N_S}$ be the small
jobs in non-increasing order of $w_j/p_j$, with ties broken by
increasing job index.
For $h=0,\ldots,N_S$, let $P_S(h)=\sum_{r=1}^{h}p_{s_r}$, so that
$P_S(0)=0$ and $P_S(N_S)=P(\mathcal J_S)$.
The following lemma shows that a prefix of this list has the largest
total weight among all sets of small jobs with no larger total
processing time.

\begin{lemma}
    \label{lem:mctcm-density}
    Let $X=\{s_1,\ldots,s_h\}$ for some $h\in\{0,\ldots,N_S\}$.
    Every $Y\subseteq\mathcal J_S$ with $P(Y)\le P(X)$ satisfies
    $W(Y)\le W(X)$.
\end{lemma}

\begin{proof}
    If $h=0$, then $P(Y)\le0$, and $Y=\varnothing$ because all
    processing times are positive.
    Let $h\ge1$ and $\rho=w_{s_h}/p_{s_h}\ge0$.
    Every job in $X\setminus Y$ has ratio $w_j/p_j\ge\rho$, and every
    job in $Y\setminus X$ has ratio $w_j/p_j\le\rho$.
    Therefore,
    \[
        \begin{aligned}
            W(X)-W(Y)
             & =W(X\setminus Y)-W(Y\setminus X)
            \ge\rho\bigl(P(X\setminus Y)-P(Y\setminus X)\bigr) \\
             & =\rho\bigl(P(X)-P(Y)\bigr)\ge0.
        \end{aligned}
        \qedhere
    \]
\end{proof}

The small jobs are assigned in list order according to cumulative
processing time bounds.
For $x\ge0$, let
\begin{equation}
    \theta(x)=\min\bigl\{h\in\{0,\ldots,N_S\}:
    P_S(h)\ge\min\{x,P(\mathcal J_S)\}\bigr\},
    \label{eq:mctcm-prefix-endpoint}
\end{equation}
which is well defined because $P_S(N_S)=P(\mathcal J_S)$.
Thus $\theta(0)=0$, $\theta(x)=N_S$ for $x\ge P(\mathcal J_S)$, and
$\theta$ is non-decreasing.
Let $\xi_0=0$ and $\xi_i=\sum_{k=1}^{i}R_k$.
Machine $M_i$ receives the small jobs
\begin{equation}
    S_i=\{s_r:\theta(\xi_{i-1})<r\le\theta(\xi_i)\},
    \qquad i\in[m].
    \label{eq:mctcm-small-assignment}
\end{equation}
Equivalently, the machines are considered in the order
$M_1,\ldots,M_m$, and $M_i$ receives the next jobs of the list until the
small jobs on $M_1,\ldots,M_i$ have total processing time at least
$\xi_i$ or the list is exhausted.
The stopping rule thus compares the cumulative load of small jobs on
$M_1,\ldots,M_i$ with $\xi_i$, rather than $P(S_i)$ with $R_i$.
In particular, $S_i$ is empty if the small jobs on
$M_1,\ldots,M_{i-1}$ already have total processing time at least
$\xi_i$.
As for large jobs, $S_i$ depends only on
$\boldsymbol q_1,\ldots,\boldsymbol q_i$.

\begin{lemma}
    \label{lem:mctcm-cumulative-filling}
    Let $R_1,\ldots,R_m\ge0$, and let $S_1,\ldots,S_m$ be given by
    \eqref{eq:mctcm-small-assignment}.
    Then $S_1,\ldots,S_m$ are pairwise disjoint, and for every
    $i\in[m]$,
    \begin{equation}
        P(S_i)\le R_i+\tau
        \quad\text{and}\quad
        \sum_{k=1}^{i}P(S_k)\ge
        \min\Bigl\{\sum_{k=1}^{i}R_k,\,P(\mathcal J_S)\Bigr\}.
        \label{eq:mctcm-cumulative-bounds}
    \end{equation}
    Moreover, if $\sum_{i=1}^{m}R_i\ge P(\mathcal J_S)$, then
    $\bigcup_{i=1}^{m}S_i=\mathcal J_S$.
\end{lemma}

\begin{proof}
    Since $\xi_0\le\xi_1\le\cdots\le\xi_m$ and $\theta$ is
    non-decreasing, the index intervals in
    \eqref{eq:mctcm-small-assignment} are pairwise disjoint, and
    $\bigcup_{k=1}^{i}S_k=\{s_1,\ldots,s_{\theta(\xi_i)}\}$.
    For $i=0,\ldots,m$, let $x_i=\min\{\xi_i,P(\mathcal J_S)\}$ and
    $y_i=P_S(\theta(\xi_i))=\sum_{k=1}^{i}P(S_k)$.
    By the definition of $\theta$, $y_i\ge x_i$, with equality if
    $x_i\in\{0,P(\mathcal J_S)\}$.
    If $0<x_i<P(\mathcal J_S)$, the minimality of $\theta(\xi_i)$ gives
    $P_S(\theta(\xi_i)-1)<x_i$, and hence
    $y_i-x_i<p_{s_{\theta(\xi_i)}}\le\tau$.
    Thus $x_i\le y_i\le x_i+\tau$ for every $i$.
    Furthermore, $x_i-x_{i-1}\le R_i$ for $i\in[m]$: if
    $\xi_{i-1}\ge P(\mathcal J_S)$, then $x_i=x_{i-1}$; otherwise,
    $x_i-x_{i-1}\le\xi_i-\xi_{i-1}=R_i$.
    Consequently, $P(S_i)=y_i-y_{i-1}\le x_i+\tau-x_{i-1}\le R_i+\tau$
    and $\sum_{k=1}^{i}P(S_k)=y_i\ge x_i$, which proves
    \eqref{eq:mctcm-cumulative-bounds}.
    Finally, if $\xi_m\ge P(\mathcal J_S)$, then $\theta(\xi_m)=N_S$,
    and every small job is assigned.
\end{proof}

\subsection{Canonical schedules}
\label{subsec:sec3-canonical}

For a configuration sequence $\boldsymbol q_1,\ldots,\boldsymbol q_m$,
let $L_i$ and $S_i$ be given by \eqref{eq:mctcm-large-assignment} and
\eqref{eq:mctcm-small-assignment} with $R_i=H-\bar P(\boldsymbol q_i)$,
and let $B_i=L_i\cup S_i$.
We call $\sigma=(B_1,\ldots,B_m)$ the canonical schedule of the
sequence.
We first verify that $\sigma$ is a schedule.
By \eqref{eq:mctcm-rounding-bound},
$\sum_{i=1}^{m}\bar P(\boldsymbol q_i)=\bar P(\boldsymbol N)
\le(1+\delta)P(\mathcal J_L)$, and
$P(\mathcal J_L)+P(\mathcal J_S)\le mT$ by assumption.
Hence
\begin{equation}
    \begin{aligned}
        \sum_{i=1}^{m}R_i
         & =mH-\bar P(\boldsymbol N)
        \ge(1+\delta)\bigl(mT-P(\mathcal J_L)\bigr) \\
         & \ge(1+\delta)P(\mathcal J_S)\ge P(\mathcal J_S).
    \end{aligned}
    \label{eq:mctcm-total-small-capacity}
\end{equation}
By Lemma~\ref{lem:mctcm-cumulative-filling}, the sets
$S_1,\ldots,S_m$ are pairwise disjoint with union $\mathcal J_S$, so
$\sigma$ assigns every job exactly once.
Moreover, $P(L_i)\le\bar P(\boldsymbol q_i)$ and
$P(S_i)\le R_i+\tau$ give
\begin{equation}
    P(B_i)\le\bar P(\boldsymbol q_i)+R_i+\tau
    =H+\tau=(1+2\delta)T=(1+\varepsilon)T,
    \qquad i\in[m].
    \label{eq:mctcm-output-load}
\end{equation}
This bound holds for every configuration sequence, regardless of
whether $\mathcal F(T)$ is empty.
A canonical schedule need not meet the deadline $T$.

Let $\mathcal C$ denote the set of canonical schedules of all
configuration sequences.
The following lemma is the key step of the analysis: every schedule
that meets the deadline can be transformed into a canonical schedule
without increasing its total assignment cost.

\begin{lemma}
    \label{lem:mctcm-reference-transformation}
    For every $\sigma^\circ=(B_1^\circ,\ldots,B_m^\circ)\in\mathcal F(T)$,
    there exists $\sigma'=(B_1',\ldots,B_m')\in\mathcal C$ such that
    \begin{equation}
        \sum_{i=1}^{r}W(B_i')\ge\sum_{i=1}^{r}W(B_i^\circ),
        \qquad r\in[m].
        \label{eq:mctcm-reference-prefix}
    \end{equation}
    Consequently, $Z(\sigma')\le Z(\sigma^\circ)$.
\end{lemma}

\begin{proof}
    Let $L_i^\circ=B_i^\circ\cap\mathcal J_L$ and
    $S_i^\circ=B_i^\circ\cap\mathcal J_S$, and let
    $q_{i,t}=|L_i^\circ\cap\mathcal J_t|$.
    By Section~\ref{subsec:sec3-rounding},
    $\boldsymbol q_1,\ldots,\boldsymbol q_m$ is a configuration
    sequence.
    Let $\sigma'$ be its canonical schedule, with
    $B_i'=L_i'\cup S_i'$.
    By Lemma~\ref{lem:mctcm-large-order},
    $\sum_{i=1}^{r}W(L_i')\ge\sum_{i=1}^{r}W(L_i^\circ)$ for every
    $r\in[m]$.

    For the small jobs, \eqref{eq:mctcm-rounding-bound} and
    $P(L_i^\circ)+P(S_i^\circ)\le T$ give
    \[
        R_i=H-\bar P(\boldsymbol q_i)
        \ge(1+\delta)\bigl(T-P(L_i^\circ)\bigr)
        \ge(1+\delta)P(S_i^\circ)\ge P(S_i^\circ),
        \qquad i\in[m].
    \]
    Fix $r\in[m]$, and let $X=\bigcup_{i=1}^{r}S_i'$ and
    $Y=\bigcup_{i=1}^{r}S_i^\circ$.
    Then $X=\{s_1,\ldots,s_{\theta(\xi_r)}\}$ is a prefix of the
    small-job list.
    The inequalities above give $P(Y)\le\xi_r$, and $P(Y)\le
    P(\mathcal J_S)$ holds trivially; together with
    Lemma~\ref{lem:mctcm-cumulative-filling}, this yields
    $P(Y)\le\min\{\xi_r,P(\mathcal J_S)\}\le P(X)$.
    Lemma~\ref{lem:mctcm-density} now gives $W(X)\ge W(Y)$, that is,
    $\sum_{i=1}^{r}W(S_i')\ge\sum_{i=1}^{r}W(S_i^\circ)$.
    Adding the inequalities for large and small jobs gives
    \eqref{eq:mctcm-reference-prefix}, and
    Corollary~\ref{cor:prelim-prefix-weight} gives
    $Z(\sigma')\le Z(\sigma^\circ)$.
\end{proof}

Applying Lemma~\ref{lem:mctcm-reference-transformation} to a schedule
that attains $Z^*(T)$ shows that, whenever $\mathcal F(T)\ne\varnothing$,
\begin{equation}
    \mathcal C\ne\varnothing
    \quad\text{and}\quad
    \min_{\sigma\in\mathcal C}Z(\sigma)\le Z^*(T).
    \label{eq:mctcm-canonical-optimum}
\end{equation}
In view of \eqref{eq:mctcm-output-load}, it therefore suffices to find
a canonical schedule of minimum total assignment cost.

\subsection{Dynamic programming}
\label{subsec:sec3-dp}

The dynamic program chooses the configurations of $M_1,\ldots,M_m$ in
this order.
Let $\mathcal Z=\{\boldsymbol z\in\mathbb Z_{\ge0}^K:
\boldsymbol z\le\boldsymbol N\}$.
A state $(i,\boldsymbol z)$ with $i\in\{0,\ldots,m\}$ and
$\boldsymbol z\in\mathcal Z$ represents the partial schedules in which
the jobs on $M_1,\ldots,M_i$ are assigned by
\eqref{eq:mctcm-large-assignment} and
\eqref{eq:mctcm-small-assignment} for some configurations
$\boldsymbol q_1,\ldots,\boldsymbol q_i\in\mathcal Q$ with
$\sum_{k=1}^{i}\boldsymbol q_k=\boldsymbol z$, and no job is assigned
to $M_{i+1},\ldots,M_m$.

All partial schedules represented by the same state assign the same
set of jobs.
By \eqref{eq:mctcm-large-assignment}, their large jobs are the first
$z_t$ jobs of each list $\ell_{t,1},\ldots,\ell_{t,N_t}$.
For the small jobs, let $A_i(\boldsymbol z)=iH-\bar P(\boldsymbol z)$.
For configurations with sum $\boldsymbol z$,
\[
    \xi_i=\sum_{k=1}^{i}\bigl(H-\bar P(\boldsymbol q_k)\bigr)
    =iH-\bar P(\boldsymbol z)=A_i(\boldsymbol z),
\]
which depends only on $i$ and $\boldsymbol z$.
Since $\xi_i\ge0$, such configurations can exist only if
$A_i(\boldsymbol z)\ge0$.
In this case, by \eqref{eq:mctcm-small-assignment}, the machines
$M_1,\ldots,M_i$ receive exactly the first
\begin{equation}
    h_i(\boldsymbol z)=\theta\bigl(A_i(\boldsymbol z)\bigr)
    \label{eq:mctcm-small-endpoint}
\end{equation}
small jobs of the list.
The machine index and the cumulative numbers of large jobs from each
class thus determine the assigned small jobs, and the dynamic program
requires no state component for small jobs.

By \eqref{eq:prelim-cost}, the total assignment cost of a schedule
equals $\sum_{i=1}^{m}a_iW(B_i)+C_0$, where $C_0=\sum_{j}b_j$ is the
same for all schedules.
For a state $(i,\boldsymbol z)$, let $D_i(\boldsymbol z)$ be the
minimum of $\sum_{k=1}^{i}a_kW(B_k)$ over the partial schedules
represented by the state, with $D_i(\boldsymbol z)=+\infty$ if the
state represents no partial schedule.

Let $(i-1,\boldsymbol y)$ be a state with
$D_{i-1}(\boldsymbol y)<+\infty$, and let $\boldsymbol q\in\mathcal Q$
with $\boldsymbol y+\boldsymbol q\le\boldsymbol N$.
Then $A_{i-1}(\boldsymbol y)\ge0$ and
$A_i(\boldsymbol y+\boldsymbol q)
    =A_{i-1}(\boldsymbol y)+H-\bar P(\boldsymbol q)
    \ge A_{i-1}(\boldsymbol y)$,
so $h_i(\boldsymbol y+\boldsymbol q)\ge h_{i-1}(\boldsymbol y)$.
Choosing $\boldsymbol q$ as the configuration of $M_i$ assigns to $M_i$
the next $q_t$ jobs of each class list and the small jobs in positions
$h_{i-1}(\boldsymbol y)+1,\ldots,h_i(\boldsymbol y+\boldsymbol q)$,
that is, the set
\begin{equation}
    \begin{aligned}
        G_i(\boldsymbol y,\boldsymbol q)
        = & \bigcup_{t=1}^{K}\{\ell_{t,r}:y_t<r\le y_t+q_t\}       \\
          & \cup\{s_r:h_{i-1}(\boldsymbol y)<r
        \le h_i(\boldsymbol y+\boldsymbol q)\}.
    \end{aligned}
    \label{eq:mctcm-new-jobs}
\end{equation}
With the prefix sums $F_t(r)=\sum_{v=1}^{r}w_{\ell_{t,v}}$ for
$r=0,\ldots,N_t$ and $F_S(h)=\sum_{v=1}^{h}w_{s_v}$ for
$h=0,\ldots,N_S$, the weight of this set is
\begin{equation}
    \begin{aligned}
        W\bigl(G_i(\boldsymbol y,\boldsymbol q)\bigr)
        = & \sum_{t=1}^{K}\bigl[F_t(y_t+q_t)-F_t(y_t)\bigr] \\
          & +F_S\bigl(h_i(\boldsymbol y+\boldsymbol q)\bigr)
        -F_S\bigl(h_{i-1}(\boldsymbol y)\bigr),
    \end{aligned}
    \label{eq:mctcm-transition-cost-prefix}
\end{equation}
which can be evaluated in $O(K)$ time.

Set $D_0(\boldsymbol0)=0$ and $D_0(\boldsymbol z)=+\infty$ for
$\boldsymbol z\ne\boldsymbol0$.
For $i=1,\ldots,m$ and $\boldsymbol z\in\mathcal Z$, set
$D_i(\boldsymbol z)=+\infty$ if $A_i(\boldsymbol z)<0$, and otherwise
\begin{equation}
    D_i(\boldsymbol z)
    =\min_{\substack{
            \boldsymbol q\in\mathcal Q,\ \boldsymbol q\le\boldsymbol z\\
            D_{i-1}(\boldsymbol z-\boldsymbol q)<+\infty}}
    \Bigl\{
    D_{i-1}(\boldsymbol z-\boldsymbol q)
    +a_iW\bigl(G_i(\boldsymbol z-\boldsymbol q,\boldsymbol q)\bigr)
    \Bigr\}.
    \label{eq:mctcm-dp-recurrence}
\end{equation}
The condition $D_{i-1}(\boldsymbol z-\boldsymbol q)<+\infty$ ensures
that $h_{i-1}(\boldsymbol z-\boldsymbol q)$ is defined, so every term
in \eqref{eq:mctcm-dp-recurrence} is well defined.
For each finite $D_i(\boldsymbol z)$, a minimizing configuration is
recorded as $\operatorname{Prev}(i,\boldsymbol z)$, with ties broken by
a fixed enumeration order of $\mathcal Q$.

By \eqref{eq:mctcm-total-small-capacity},
$A_m(\boldsymbol N)\ge P(\mathcal J_S)$, and hence
$h_m(\boldsymbol N)=N_S$.
If $D_m(\boldsymbol N)<+\infty$, a schedule
$\widehat\sigma=(\widehat B_1,\ldots,\widehat B_m)$ is recovered as
follows: starting from $\boldsymbol z=\boldsymbol N$, for
$i=m,m-1,\ldots,1$, let
$\boldsymbol q=\operatorname{Prev}(i,\boldsymbol z)$, set
$\widehat B_i=G_i(\boldsymbol z-\boldsymbol q,\boldsymbol q)$, and
replace $\boldsymbol z$ by $\boldsymbol z-\boldsymbol q$.

\begin{lemma}
    \label{lem:mctcm-dp-optimality}
    For all $i\in\{0,\ldots,m\}$ and $\boldsymbol z\in\mathcal Z$, the
    value $D_i(\boldsymbol z)$ computed by
    \eqref{eq:mctcm-dp-recurrence} equals the minimum of
    $\sum_{k=1}^{i}a_kW(B_k)$ over the partial schedules represented by
    $(i,\boldsymbol z)$, and $D_i(\boldsymbol z)=+\infty$ if the state
    represents no partial schedule.
    In particular, $D_m(\boldsymbol N)+C_0=\min_{\sigma\in\mathcal C}
    Z(\sigma)$.
    If $D_m(\boldsymbol N)<+\infty$, the recovered schedule
    $\widehat\sigma$ belongs to $\mathcal C$ and satisfies
    $Z(\widehat\sigma)=D_m(\boldsymbol N)+C_0$.
\end{lemma}

\begin{proof}
    We proceed by induction on $i$.
    For $i=0$, the only represented partial schedule is the empty one,
    with count vector $\boldsymbol0$ and cost zero.
    Let $i\ge1$, and assume that the claim holds for $i-1$.
    If $A_i(\boldsymbol z)<0$, no configurations
    $\boldsymbol q_1,\ldots,\boldsymbol q_i\in\mathcal Q$ sum to
    $\boldsymbol z$, because any such configurations would satisfy
    $\bar P(\boldsymbol z)=\sum_{k=1}^{i}\bar P(\boldsymbol q_k)\le iH$;
    hence $D_i(\boldsymbol z)=+\infty$ is correct.
    Let $A_i(\boldsymbol z)\ge0$.

    Consider the term of \eqref{eq:mctcm-dp-recurrence} for a
    configuration $\boldsymbol q$, and let
    $\boldsymbol y=\boldsymbol z-\boldsymbol q$.
    By the induction hypothesis, some partial schedule represented by
    $(i-1,\boldsymbol y)$ has cost $D_{i-1}(\boldsymbol y)$; it assigns
    the first $y_t$ jobs of each class list and the first
    $h_{i-1}(\boldsymbol y)$ small jobs.
    Assigning $G_i(\boldsymbol y,\boldsymbol q)$ to $M_i$ extends it in
    accordance with both assignment rules, so the result is a partial
    schedule represented by $(i,\boldsymbol z)$ whose cost equals the
    term.
    Conversely, consider a partial schedule represented by
    $(i,\boldsymbol z)$ and generated by
    $\boldsymbol q_1,\ldots,\boldsymbol q_i$, and let
    $\boldsymbol q=\boldsymbol q_i$ and
    $\boldsymbol y=\boldsymbol z-\boldsymbol q$.
    Since $L_k$ and $S_k$ depend only on
    $\boldsymbol q_1,\ldots,\boldsymbol q_k$, its restriction to
    $M_1,\ldots,M_{i-1}$ is represented by $(i-1,\boldsymbol y)$ and,
    by the induction hypothesis, has cost at least
    $D_{i-1}(\boldsymbol y)<+\infty$.
    Machine $M_i$ receives exactly $G_i(\boldsymbol y,\boldsymbol q)$,
    so the cost of the partial schedule is at least the term of
    \eqref{eq:mctcm-dp-recurrence} for $\boldsymbol q$.
    The two directions prove the claim for $(i,\boldsymbol z)$; in
    particular, the minimum in \eqref{eq:mctcm-dp-recurrence} is taken
    over an empty set exactly when the state represents no partial
    schedule.

    The partial schedules represented by $(m,\boldsymbol N)$ are
    exactly the canonical schedules, because configurations with sum
    $\boldsymbol N$ form configuration sequences and
    $h_m(\boldsymbol N)=N_S$.
    Together with \eqref{eq:prelim-cost}, this gives
    $D_m(\boldsymbol N)+C_0=\min_{\sigma\in\mathcal C}Z(\sigma)$.
    If $D_m(\boldsymbol N)<+\infty$, each recorded configuration
    attains the minimum in \eqref{eq:mctcm-dp-recurrence} at a state
    with finite value, so the recovery passes only through such states
    and ends at $(0,\boldsymbol0)$, the only state of layer $0$ with
    finite value.
    The recorded configurations therefore form a configuration sequence
    whose canonical schedule is $\widehat\sigma$, and the costs of the
    corresponding transitions sum to $D_m(\boldsymbol N)$.
\end{proof}

\subsection{The algorithm and its analysis}
\label{subsec:sec3-algorithm}

The complete algorithm is summarized below.

\begin{description}
    \item[Algorithm $\mathcal A_1$.]
    \item[Input:] An instance $(I,T)$ of
          $P|C_{\max}\le T|\operatorname{TMC}$ with
          $0\le a_1\le\cdots\le a_m$, and an accuracy parameter
          $0<\varepsilon\le1$.
    \item[Output:] A schedule $\widehat\sigma$ together with
          $\widehat C=C_{\max}(\widehat\sigma)$ and
          $\widehat Z=Z(\widehat\sigma)$, or $\mathrm{FAIL}$.
    \item[Step 1.] If $p_{\max}>T$ or $P(\mathcal J)>mT$, return
          $\mathrm{FAIL}$.
    \item[Step 2.] Set $\delta=\varepsilon/2$, $\tau=\delta T$, and
          $H=(1+\delta)T$.
          Determine $\mathcal J_L$, $\mathcal J_S$, the values
          $\pi_0,\ldots,\pi_K$, the classes
          $\mathcal J_1,\ldots,\mathcal J_K$ in
          \eqref{eq:mctcm-job-classes}, and $\boldsymbol N$.
    \item[Step 3.] Sort each nonempty class $\mathcal J_t$ in
          non-increasing order of $w_j$ and $\mathcal J_S$ in
          non-increasing order of $w_j/p_j$, breaking ties by
          increasing job index.
          Compute the prefix sums $P_S$, $F_S$, and
          $F_1,\ldots,F_K$.
    \item[Step 4.] Construct $\mathcal Q$ by enumerating all
          $\boldsymbol q\in\mathbb Z_{\ge0}^K$ with
          $\sum_{t}q_t\le\kappa$ and retaining those with
          $\boldsymbol q\le\boldsymbol N$ and
          $\bar P(\boldsymbol q)\le H$.
          For all $i\in\{0,\ldots,m\}$ and $\boldsymbol z\in\mathcal Z$
          with $A_i(\boldsymbol z)\ge0$, compute $h_i(\boldsymbol z)$
          by binary search on $P_S(0),\ldots,P_S(N_S)$.
    \item[Step 5.] Compute $D_i(\boldsymbol z)$ and
          $\operatorname{Prev}(i,\boldsymbol z)$ for
          $i=0,\ldots,m$ and $\boldsymbol z\in\mathcal Z$ by
          \eqref{eq:mctcm-dp-recurrence}, evaluating the transition
          weights by \eqref{eq:mctcm-transition-cost-prefix}.
    \item[Step 6.] If $D_m(\boldsymbol N)=+\infty$, return
          $\mathrm{FAIL}$.
          Otherwise, recover $\widehat\sigma$ from the records
          $\operatorname{Prev}$, and return $\widehat\sigma$,
          $\widehat C=\max_{i\in[m]}P(\widehat B_i)$ computed with the
          original processing times, and
          $\widehat Z=D_m(\boldsymbol N)+C_0$.
\end{description}

\begin{theorem}
    \label{thm:problem1-guarantee}
    For every $0<\varepsilon\le1$, algorithm $\mathcal A_1$ is a
    $(1+\varepsilon,1)$-approximation algorithm for
    $P|C_{\max}\le T|\operatorname{TMC}$.
    More precisely, every schedule returned by $\mathcal A_1$ satisfies
    $C_{\max}(\widehat\sigma)\le(1+\varepsilon)T$, even if
    $\mathcal F(T)=\varnothing$; if $\mathcal F(T)\ne\varnothing$, the
    algorithm returns a schedule with $Z(\widehat\sigma)\le Z^*(T)$;
    and if it returns $\mathrm{FAIL}$, then $\mathcal F(T)=\varnothing$.
    The running time of $\mathcal A_1$ is
    \[
        O\Bigl(n\log(n+1)
        +(1+m/\varepsilon)^{O(\varepsilon^{-1}\log(2/\varepsilon))}\Bigr).
    \]
\end{theorem}

\begin{proof}
    If Step~1 returns $\mathrm{FAIL}$, then $\mathcal F(T)=\varnothing$
    by Section~\ref{sec:preliminaries}.
    Otherwise, Steps~2--5 carry out the constructions of
    Sections~\ref{subsec:sec3-rounding}--\ref{subsec:sec3-dp}.
    By Lemma~\ref{lem:mctcm-dp-optimality}, Step~6 returns
    $\mathrm{FAIL}$ exactly when $\mathcal C=\varnothing$, which by
    \eqref{eq:mctcm-canonical-optimum} implies
    $\mathcal F(T)=\varnothing$.
    Otherwise, it returns a schedule $\widehat\sigma\in\mathcal C$ of
    minimum total assignment cost over $\mathcal C$, together with its
    makespan and its cost.
    By \eqref{eq:mctcm-output-load},
    $C_{\max}(\widehat\sigma)\le(1+\varepsilon)T$.
    If $\mathcal F(T)\ne\varnothing$,
    \eqref{eq:mctcm-canonical-optimum} gives
    $Z(\widehat\sigma)\le Z^*(T)$.
    Hence $\mathcal A_1$ is a $(1+\varepsilon,1)$-approximation
    algorithm.

    For the running time, let $V=|\mathcal Z|=\prod_{t=1}^{K}(N_t+1)$
    and $Q=|\mathcal Q|$.
    Since $V\ge1+|\mathcal J_L|$, Steps~1--3 take
    $O(n\log(n+1)+KV)$ time, including the classification of the large
    jobs by binary search over $\pi_0,\ldots,\pi_K$.
    In Step~4, the enumeration takes $O(K\binom{K+\kappa}{K})$ time,
    the values $\bar P(\boldsymbol z)$ take $O(KV)$ time, and the
    binary searches take $O(mV(K+\log(n+1)))$ time.
    Step~5 considers at most $mVQ$ pairs of a state and a configuration,
    each in $O(K)$ time, and Step~6 takes $O(mK+n)$ time.
    The total running time is therefore
    \begin{equation}
        O\Bigl(n\log(n+1)+K\tbinom{K+\kappa}{K}
        +mV\bigl(KQ+\log(n+1)\bigr)\Bigr).
        \label{eq:mctcm-detailed-time}
    \end{equation}
    Since $N_t\le|\mathcal J_L|<m/\delta$, we have
    $V\le(1+m/\delta)^K$.
    Moreover, $Q\le\binom{K+\kappa}{K}\le(\kappa+1)^K$ with
    $\kappa+1\le2+1/\delta\le(1+m/\delta)^2$, $K\le2^K$, and
    $m\le1+m/\delta$.
    Finally, $mV\log(n+1)\le(mV)^2+\log^2(n+1)$ and
    $\log^2(n+1)\le n\log(n+1)$.
    Hence all terms of \eqref{eq:mctcm-detailed-time} other than
    $n\log(n+1)$ are bounded by $(1+m/\delta)^{O(K)}$.
    Substituting $\delta=\varepsilon/2$ and using
    $1+2m/\varepsilon\le(1+m/\varepsilon)^2$ and
    $K=O(\varepsilon^{-1}\log(2/\varepsilon))$ yields the stated bound.
\end{proof}

\section{A PTAS for $P|\operatorname{TMC}\le\widehat U|C_{\max}$}
\label{sec:ptas-budget}

Throughout this section and the next, $\mathcal A_1(I,T,\eta)$ denotes
a call of algorithm $\mathcal A_1$ on the instance $(I,T)$ with
accuracy parameter $\eta$.
By Theorem~\ref{thm:problem1-guarantee}, every schedule returned by
such a call has makespan at most $(1+\eta)T$, even if
$\mathcal F(T)=\varnothing$, and the call returns a schedule with total
assignment cost at most $Z^*(T)$ whenever $\mathcal F(T)\ne\varnothing$.

By Section~\ref{sec:preliminaries}, a schedule satisfying the budget
exists if and only if $a_1W(\mathcal J)+C_0\le\widehat U$, and in this
case $L\le C^*(\widehat U)\le P(\mathcal J)$, where
$L=\max\{p_{\max},P(\mathcal J)/m\}$.
The algorithm is based on the following observation.
If $T\ge C^*(\widehat U)$, then a schedule attaining $C^*(\widehat U)$
belongs to $\mathcal F(T)$ and has total assignment cost at most
$\widehat U$, so that
\begin{equation}
    \mathcal F(T)\ne\varnothing
    \quad\text{and}\quad
    Z^*(T)\le\widehat U,
    \qquad T\ge C^*(\widehat U).
    \label{eq:budget-deadline}
\end{equation}
Hence $\mathcal A_1(I,T,\eta)$ returns a schedule that satisfies the
budget whenever $T\ge C^*(\widehat U)$.
Since $C^*(\widehat U)$ is not known, the algorithm tries the deadlines
of a geometric sequence in increasing order and accepts the first
returned schedule that satisfies the budget.

\begin{description}
    \item[Algorithm $\mathcal A_2$.]
    \item[Input:] An instance $(I,\widehat U)$ of
          $P|\operatorname{TMC}\le\widehat U|C_{\max}$ with
          $0\le a_1\le\cdots\le a_m$, and an accuracy parameter
          $0<\varepsilon\le1$.
    \item[Output:] A schedule $\widehat\sigma$ together with
          $\widehat C=C_{\max}(\widehat\sigma)$ and
          $\widehat Z=Z(\widehat\sigma)$, or a report that no schedule
          satisfies the budget.
    \item[Step 1.] If $a_1W(\mathcal J)+C_0>\widehat U$, report that no
          schedule satisfies the budget and stop.
    \item[Step 2.] Set $\eta=\varepsilon/4$ and
          $L=\max\{p_{\max},P(\mathcal J)/m\}$.
          Let $T_r=(1+\eta)^rL$ for $r=0,\ldots,r_{\max}$, where
          $r_{\max}$ is the smallest nonnegative integer with
          $T_{r_{\max}}\ge P(\mathcal J)$.
    \item[Step 3.] For $r=0,1,\ldots,r_{\max}$, call
          $\mathcal A_1(I,T_r,\eta)$.
          If the call returns a schedule $\sigma^r$ with
          $Z(\sigma^r)\le\widehat U$, return $\sigma^r$ together with
          its makespan and total assignment cost, and stop.
\end{description}

\begin{theorem}
    \label{thm:problem2-guarantee}
    For every $0<\varepsilon\le1$, algorithm $\mathcal A_2$ is a
    $(1,1+\varepsilon)$-approximation algorithm for
    $P|\operatorname{TMC}\le\widehat U|C_{\max}$ with running time
    \[
        O\Bigl(\bigl[1+\varepsilon^{-1}\log(\min\{m,n\})\bigr]
        n\log(n+1)
        +(1+m/\varepsilon)^{O(\varepsilon^{-1}\log(2/\varepsilon))}\Bigr).
    \]
    Consequently, the algorithms $\mathcal A_2$ with
    $0<\varepsilon\le1$ form a PTAS for this problem.
\end{theorem}

\begin{proof}
    If Step~1 stops, no schedule satisfies the budget by
    Section~\ref{sec:preliminaries}.
    Suppose that $a_1W(\mathcal J)+C_0\le\widehat U$, so that
    $L\le C^*(\widehat U)\le P(\mathcal J)$.
    Since $T_0=L$ and $T_{r_{\max}}\ge P(\mathcal J)$, there is a
    smallest index $r^*\le r_{\max}$ with $T_{r^*}\ge C^*(\widehat U)$.
    Either $r^*=0$ and $T_{r^*}=L\le C^*(\widehat U)$, or
    $T_{r^*-1}<C^*(\widehat U)$; in both cases,
    \begin{equation}
        C^*(\widehat U)\le T_{r^*}\le(1+\eta)C^*(\widehat U).
        \label{eq:budget-rstar}
    \end{equation}
    By \eqref{eq:budget-deadline} and
    Theorem~\ref{thm:problem1-guarantee}, the call
    $\mathcal A_1(I,T_{r^*},\eta)$ returns a schedule with total
    assignment cost at most $Z^*(T_{r^*})\le\widehat U$.
    Step~3 therefore returns a schedule $\widehat\sigma$ at some index
    $r_0\le r^*$, and $Z(\widehat\sigma)\le\widehat U$ by the acceptance
    condition.
    By Theorem~\ref{thm:problem1-guarantee}, \eqref{eq:budget-rstar},
    and $r_0\le r^*$,
    \[
        C_{\max}(\widehat\sigma)
        \le(1+\eta)T_{r_0}
        \le(1+\eta)T_{r^*}
        \le(1+\eta)^2C^*(\widehat U).
    \]
    Since $\eta=\varepsilon/4$ and $0<\varepsilon\le1$,
    \[
        (1+\eta)^2=1+\frac{\varepsilon}{2}+\frac{\varepsilon^2}{16}
        \le1+\frac{9\varepsilon}{16}\le1+\varepsilon.
    \]
    Hence $C_{\max}(\widehat\sigma)\le(1+\varepsilon)C^*(\widehat U)$,
    and $\mathcal A_2$ is a $(1,1+\varepsilon)$-approximation
    algorithm.

    For the running time, the inequalities
    $\log(1+\eta)\ge\eta/(1+\eta)$ and
    $P(\mathcal J)/L\le\min\{m,n\}$ give
    \begin{equation}
        r_{\max}+1
        =1+\left\lceil\frac{\log(P(\mathcal J)/L)}{\log(1+\eta)}
        \right\rceil
        \le2+\frac{5}{\varepsilon}\log(\min\{m,n\}).
        \label{eq:tccmm-number-of-calls}
    \end{equation}
    Steps~1 and~2 take $O(n+r_{\max})$ time, and Step~3 makes at most
    $r_{\max}+1$ calls of $\mathcal A_1$ with accuracy $\eta$.
    By Theorem~\ref{thm:problem1-guarantee}, each call takes
    $O(n\log(n+1)+(1+m/\eta)^{O(\eta^{-1}\log(2/\eta))})$ time.
    Since $1+m/\eta=1+4m/\varepsilon\le(1+m/\varepsilon)^3$ and
    $\eta^{-1}\log(2/\eta)=O(\varepsilon^{-1}\log(2/\varepsilon))$,
    this is
    $O(n\log(n+1)+(1+m/\varepsilon)^{O(\varepsilon^{-1}
    \log(2/\varepsilon))})$.
    Multiplying the first term by the bound in
    \eqref{eq:tccmm-number-of-calls} gives the first term of the
    stated running time.
    For the second term, $r_{\max}+1\le2+5m/\varepsilon
    \le5(1+m/\varepsilon)$, so the factor $r_{\max}+1$ changes only the
    constant in the exponent.
    Finally, the running time is polynomial in $|I|$ for each fixed
    $\varepsilon$, and $\mathcal A_2$ never violates the budget, so the
    algorithms $\mathcal A_2$ form a PTAS.
\end{proof}

\section{An approximation algorithm for
  $\operatorname{LEX}(C_{\max},\operatorname{TMC})$}
\label{sec:lex-approximation}

This section uses $\mathcal A_1$ to obtain, for each fixed
$0<\varepsilon\le1$, a polynomial-time
$(1+\varepsilon,1)$-approximation algorithm for
$\operatorname{LEX}(C_{\max},\operatorname{TMC})$.
Recall that $C^*$ is the minimum makespan and that $Z^*=Z^*(C^*)$.
Every $T\ge C^*$ satisfies $\mathcal F(T)\ne\varnothing$ and, by
\eqref{eq:prelim-monotone}, $Z^*(T)\le Z^*$.
Consequently, a call $\mathcal A_1(I,T,\eta)$ with $T\ge C^*$ returns
a schedule with total assignment cost at most $Z^*$.
In contrast to the budget $\widehat U$ in Section~\ref{sec:ptas-budget},
however, the bound $Z^*$ is not part of the input, so the acceptance
test of $\mathcal A_2$ is not available.

The algorithm therefore proceeds in two phases.
It first calls $\mathcal A_1$ with the deadlines $T_r$ of
Section~\ref{sec:ptas-budget} in increasing order and stops at the
first call that returns a schedule $\sigma^0$.
The deadline $T_{r_0}$ of this call may be smaller than $C^*$, in which
case $\mathcal F(T_{r_0})=\varnothing$ and the total assignment cost of
$\sigma^0$ is not controlled.
The makespan $\widehat T=C_{\max}(\sigma^0)$, however, satisfies
$\widehat T\ge C^*$ and, as shown below, exceeds $C^*$ by a factor of
at most $(1+\eta)^2$.
A second call with deadline $\widehat T$ therefore returns a schedule
whose total assignment cost is at most $Z^*$ and whose makespan is at
most $(1+\eta)\widehat T$.

\begin{description}
    \item[Algorithm $\mathcal A_3$.]
    \item[Input:] An instance $I$ of
          $\operatorname{LEX}(C_{\max},\operatorname{TMC})$ with
          $0\le a_1\le\cdots\le a_m$, and an accuracy parameter
          $0<\varepsilon\le1$.
    \item[Output:] A schedule $\widehat\sigma$ together with
          $\widehat C=C_{\max}(\widehat\sigma)$ and
          $\widehat Z=Z(\widehat\sigma)$.
    \item[Step 1.] Set $\eta=\varepsilon/4$, and define $L$, $T_r$, and
          $r_{\max}$ as in Step~2 of $\mathcal A_2$.
    \item[Step 2.] For $r=0,1,\ldots,r_{\max}$, call
          $\mathcal A_1(I,T_r,\eta)$.
          At the first call that returns a schedule $\sigma^0$, set
          $\widehat T=C_{\max}(\sigma^0)$ and go to Step~3.
    \item[Step 3.] Call $\mathcal A_1(I,\widehat T,\eta)$ and return its
          output.
\end{description}

\begin{theorem}
    \label{thm:problem3-guarantee}
    For every $0<\varepsilon\le1$, algorithm $\mathcal A_3$ is a
    $(1+\varepsilon,1)$-approximation algorithm for
    $\operatorname{LEX}(C_{\max},\operatorname{TMC})$ with running time
    \[
        O\Bigl(\bigl[1+\varepsilon^{-1}\log(\min\{m,n\})\bigr]
        n\log(n+1)
        +(1+m/\varepsilon)^{O(\varepsilon^{-1}\log(2/\varepsilon))}\Bigr).
    \]
\end{theorem}

\begin{proof}
    Since $L\le C^*\le P(\mathcal J)$, the argument that leads to
    \eqref{eq:budget-rstar}, with $C^*$ in place of $C^*(\widehat U)$,
    yields an index $r^*\le r_{\max}$ with
    $C^*\le T_{r^*}\le(1+\eta)C^*$.
    Because $\mathcal F(T_{r^*})\ne\varnothing$,
    Theorem~\ref{thm:problem1-guarantee} ensures that
    $\mathcal A_1(I,T_{r^*},\eta)$ returns a schedule.
    Hence the first call in Step~2 that returns a schedule has an index
    $r_0\le r^*$.
    Its output $\sigma^0$ is a schedule of $I$, so
    $\widehat T=C_{\max}(\sigma^0)\ge C^*$, and
    Theorem~\ref{thm:problem1-guarantee} gives
    \begin{equation}
        C^*\le\widehat T
        \le(1+\eta)T_{r_0}
        \le(1+\eta)T_{r^*}
        \le(1+\eta)^2C^*.
        \label{eq:lex-actual-time-bound}
    \end{equation}

    Since $\widehat T\ge C^*$, the set $\mathcal F(\widehat T)$
    contains a lexicographically optimal schedule, and
    \eqref{eq:prelim-monotone} gives
    $Z^*(\widehat T)\le Z^*(C^*)=Z^*$.
    By Theorem~\ref{thm:problem1-guarantee}, Step~3 returns a schedule
    $\widehat\sigma$ with $Z(\widehat\sigma)\le Z^*(\widehat T)\le Z^*$
    and, by \eqref{eq:lex-actual-time-bound},
    \[
        C_{\max}(\widehat\sigma)
        \le(1+\eta)\widehat T
        \le(1+\eta)^3C^*.
    \]
    Since $\eta=\varepsilon/4$ and $0<\varepsilon\le1$,
    \[
        (1+\eta)^3
        =1+\frac{3\varepsilon}{4}+\frac{3\varepsilon^2}{16}
        +\frac{\varepsilon^3}{64}
        \le1+\frac{61\varepsilon}{64}\le1+\varepsilon.
    \]
    Hence $C_{\max}(\widehat\sigma)\le(1+\varepsilon)C^*$ and
    $Z(\widehat\sigma)\le Z^*$, which proves that $\mathcal A_3$ is a
    $(1+\varepsilon,1)$-approximation algorithm.

    The deadline sequence and the accuracy of the calls are the same as
    in $\mathcal A_2$, and $\mathcal A_3$ makes at most $r_{\max}+2$
    calls of $\mathcal A_1$.
    The analysis in the proof of Theorem~\ref{thm:problem2-guarantee}
    therefore yields the stated running time.
\end{proof}

\section{Extension to costs $c_{ij}=s_i+a_iw_j$}
\label{sec:affine-qptas}

This section considers the cost function $c_{ij}=s_i+a_iw_j$, in which
a base cost $s_i\ge0$ is charged for each job processed on machine
$M_i$.
We write $I=(\mathcal M,\mathcal J;p,s,a,w)$ for the machine and job
data and index the machines so that $0\le a_1\le\cdots\le a_m$; no
ordering relation is imposed on $s_1,\ldots,s_m$.
All other assumptions are as in Section~\ref{sec:preliminaries}.
The total assignment cost of a schedule $\sigma=(B_1,\ldots,B_m)$ is
\begin{equation}
    Z(\sigma)=\sum_{i=1}^{m}\bigl(s_i|B_i|+a_iW(B_i)\bigr),
    \label{eq:affine-cost}
\end{equation}
and $\mathcal F(T)$, $Z^*(T)$, $C^*(\widehat U)$, $C^*$, and $Z^*$ are
defined as in Section~\ref{sec:preliminaries} with this cost function.
In particular, \eqref{eq:prelim-monotone} and
\eqref{eq:prelim-makespan-bounds} remain valid.

The base cost $\sum_is_i|B_i|$ depends on the number of jobs assigned
to each machine.
If the number of jobs from each rounded class on each machine is fixed,
reassigning jobs within a class leaves this term unchanged, so the rule
for large jobs in Section~\ref{sec3} remains applicable.
The rule for small jobs, by contrast, maximizes prefix sums of weights
but may change the number of jobs on each machine, and it may
therefore increase the base cost.
We thus treat separately only the jobs with processing time at most
$\varepsilon T/(2n)$.
Their total processing time is at most $\varepsilon T/2$, so each of
them can be assigned to a machine of minimum cost without regard to
the machine loads.
The remaining jobs are assigned by a dynamic program over rounded
classes, as in Section~\ref{sec3}.
Since their processing times can be as small as $\varepsilon T/(2n)$,
the number of classes grows to $O(\varepsilon^{-1}\log(n/\varepsilon))$
and the number of jobs on a machine is no longer bounded by a constant,
which leads to a quasi-polynomial running time.

\subsection{The deadline-constrained problem}
\label{subsec:affine-deadline}

Fix $T>0$ and $0<\varepsilon\le1$.
As in Section~\ref{sec3}, we assume that $p_{\max}\le T$ and
$P(\mathcal J)\le mT$, since otherwise $\mathcal F(T)=\varnothing$, and
we set
\[
    \delta=\varepsilon/2,
    \qquad
    \tau=\delta T/n,
    \qquad
    H=(1+\delta)T.
\]
The symbols $\tau$, $\pi_t$, $K$, $\mathcal J_t$, $\boldsymbol N$,
$\ell_{t,r}$, $F_t$, $\mathcal Z$, $D_i$, and $\operatorname{Prev}$ are
redefined in this section.

\paragraph{Tiny jobs.}
Let $\mathcal J_0=\{j\in\mathcal J:p_j\le\tau\}$ and
$\mathcal J_+=\mathcal J\setminus\mathcal J_0$.
Each job $j\in\mathcal J_0$ is assigned to machine $M_{\mu(j)}$, where
$\mu(j)$ is the smallest index in
$\arg\min_{i\in[m]}(s_i+a_iw_j)$.
Let $E_i=\{j\in\mathcal J_0:\mu(j)=i\}$ for $i\in[m]$, and let
$Z_0=\sum_{j\in\mathcal J_0}\min_{i\in[m]}(s_i+a_iw_j)$ denote the
total cost of these assignments.
Since $|\mathcal J_0|\le n$,
\begin{equation}
    P(E_i)\le P(\mathcal J_0)\le n\tau=\delta T,
    \qquad i\in[m].
    \label{eq:affine-tiny-load}
\end{equation}

\paragraph{Rounding and ordered assignment.}
Let $\pi_t=(1+\delta)^t\tau$ for $t=0,1,\ldots$, and let $K$ be the
smallest index with $\pi_K\ge T$, that is,
\[
    K=\left\lceil\frac{\log(n/\delta)}{\log(1+\delta)}\right\rceil
    =O\bigl(\delta^{-1}\log(n/\delta)\bigr).
\]
In contrast to Section~\ref{sec3}, $K$ depends on $n$; it is
$O(\log(n+1))$ for fixed $\varepsilon$.
The classes $\mathcal J_t=\{j\in\mathcal J_+:\pi_{t-1}<p_j\le\pi_t\}$,
$t\in[K]$, are pairwise disjoint and cover $\mathcal J_+$.
Let $N_t=|\mathcal J_t|$, $\boldsymbol N=(N_1,\ldots,N_K)$, and
$\bar p_j=\pi_t$ for $j\in\mathcal J_t$.
The rounding bound \eqref{eq:mctcm-rounding-bound} holds for every
subset of $\mathcal J_+$, with the same proof.
For each nonempty class, let $\ell_{t,1},\ldots,\ell_{t,N_t}$ be its
jobs in non-increasing order of weight, with ties broken by increasing
job index, and let $F_t(r)=\sum_{v=1}^{r}w_{\ell_{t,v}}$ for
$r=0,\ldots,N_t$.
For $\boldsymbol x\in\mathbb Z_{\ge0}^K$, let
$\bar P(\boldsymbol x)=\sum_{t=1}^{K}x_t\pi_t$.
The following lemma adapts Lemma~\ref{lem:mctcm-large-order} to the
present cost function.

\begin{lemma}
    \label{lem:affine-ordered-assignment}
    Let $L_1,\ldots,L_m$ be pairwise disjoint, possibly empty, subsets
    of $\mathcal J_+$ with union $\mathcal J_+$, and let
    $q_{i,t}=|L_i\cap\mathcal J_t|$ for $i\in[m]$ and $t\in[K]$.
    Let $L_1',\ldots,L_m'$ be obtained by considering
    $M_1,\ldots,M_m$ in this order and assigning to $M_i$ the next
    $q_{i,t}$ jobs of the list of $\mathcal J_t$ for every $t\in[K]$,
    as in \eqref{eq:mctcm-large-assignment}.
    Then $L_1',\ldots,L_m'$ are pairwise disjoint with union
    $\mathcal J_+$, $|L_i'\cap\mathcal J_t|=q_{i,t}$ for all $i$ and
    $t$, and
    \[
        \sum_{i=1}^{m}\bigl(s_i|L_i'|+a_iW(L_i')\bigr)
        \le\sum_{i=1}^{m}\bigl(s_i|L_i|+a_iW(L_i)\bigr).
    \]
    In particular, $\sum_{j\in L_i'}\bar p_j=\sum_{j\in L_i}\bar p_j$
    for every $i\in[m]$.
\end{lemma}

\begin{proof}
    Since $\sum_{i=1}^{m}q_{i,t}=N_t$ for every $t$, every job of
    $\mathcal J_+$ is assigned exactly once, and
    $|L_i'\cap\mathcal J_t|=q_{i,t}$.
    Hence $|L_i'|=|L_i|$, so the base costs coincide on every machine,
    and both $L_i'$ and $L_i$ have rounded load
    $\sum_{t=1}^{K}q_{i,t}\pi_t$.
    The argument in the proof of Lemma~\ref{lem:mctcm-large-order}
    shows that $\sum_{i=1}^{r}W(L_i')\ge\sum_{i=1}^{r}W(L_i)$ for every
    $r\in[m]$, and Lemma~\ref{lem:mctcm-prefix-cost} then gives
    $\sum_{i=1}^{m}a_iW(L_i')\le\sum_{i=1}^{m}a_iW(L_i)$.
\end{proof}

\paragraph{Dynamic program.}
Let $\mathcal Z=\{\boldsymbol v\in\mathbb Z_{\ge0}^K:
\boldsymbol v\le\boldsymbol N\}$.
A state $(i,\boldsymbol v)$ with $i\in\{0,\ldots,m\}$ and
$\boldsymbol v\in\mathcal Z$ represents the assignments of the first
$v_t$ jobs of each list to $M_1,\ldots,M_i$ in which, for every $t$,
the machines $M_1,\ldots,M_i$ receive consecutive blocks of the list
of $\mathcal J_t$ in this order, and the rounded load of every machine
is at most $H$.
In contrast to Section~\ref{sec3}, no configuration set is used: a
machine may receive any number of jobs, subject only to the bound $H$.
Let $D_i(\boldsymbol v)$ be the minimum of
$\sum_{k=1}^{i}\bigl(s_k|L_k|+a_kW(L_k)\bigr)$ over these assignments,
where $L_k$ is the set of jobs assigned to $M_k$, with
$D_i(\boldsymbol v)=+\infty$ if no such assignment exists.

If $M_i$ receives the jobs in positions $u_t+1,\ldots,v_t$ of each
list, where $\boldsymbol u\le\boldsymbol v$, its rounded load is
$\bar P(\boldsymbol v)-\bar P(\boldsymbol u)$, and its cost is
\begin{equation}
    g_i(\boldsymbol u,\boldsymbol v)
    =s_i\sum_{t=1}^{K}(v_t-u_t)
    +a_i\sum_{t=1}^{K}\bigl[F_t(v_t)-F_t(u_t)\bigr].
    \label{eq:affine-transition-cost}
\end{equation}
Set $D_0(\boldsymbol0)=0$ and $D_0(\boldsymbol v)=+\infty$ for
$\boldsymbol v\ne\boldsymbol0$.
For $i\in[m]$ and $\boldsymbol v\in\mathcal Z$, let
\begin{equation}
    D_i(\boldsymbol v)
    =\min_{\substack{
            \boldsymbol u\in\mathcal Z,\ \boldsymbol u\le\boldsymbol v\\
            \bar P(\boldsymbol v)-\bar P(\boldsymbol u)\le H\\
            D_{i-1}(\boldsymbol u)<+\infty}}
    \bigl\{D_{i-1}(\boldsymbol u)+g_i(\boldsymbol u,\boldsymbol v)\bigr\},
    \label{eq:affine-recurrence}
\end{equation}
and record a minimizing vector $\boldsymbol u$ as
$\operatorname{Prev}(i,\boldsymbol v)$ whenever $D_i(\boldsymbol v)$ is
finite.
The choice $\boldsymbol u=\boldsymbol v$ assigns no job of
$\mathcal J_+$ to $M_i$.
The induction in the proof of Lemma~\ref{lem:mctcm-dp-optimality}
shows that the values computed by \eqref{eq:affine-recurrence} are
the minima defined above.

\paragraph{The algorithm.}
The complete procedure for the deadline-constrained problem is as
follows.

\begin{description}
    \item[Algorithm $\mathcal A_1'$.]
    \item[Input:] An instance $(I,T)$ of
          $P|C_{\max}\le T|\operatorname{TMC}$ with costs
          $c_{ij}=s_i+a_iw_j$ and $0\le a_1\le\cdots\le a_m$, and an
          accuracy parameter $0<\varepsilon\le1$.
    \item[Output:] A schedule $\widehat\sigma$ together with
          $\widehat C=C_{\max}(\widehat\sigma)$ and
          $\widehat Z=Z(\widehat\sigma)$, or $\mathrm{FAIL}$.
    \item[Step 1.] If $p_{\max}>T$ or $P(\mathcal J)>mT$, return
          $\mathrm{FAIL}$.
    \item[Step 2.] Set $\delta=\varepsilon/2$, $\tau=\delta T/n$, and
          $H=(1+\delta)T$.
          Determine $\mathcal J_0$, $\mathcal J_+$, the sets
          $E_1,\ldots,E_m$, and $Z_0$.
    \item[Step 3.] Determine $\pi_0,\ldots,\pi_K$, the classes
          $\mathcal J_1,\ldots,\mathcal J_K$, and $\boldsymbol N$.
          Sort each nonempty class in non-increasing order of $w_j$,
          breaking ties by increasing job index, and compute
          $F_1,\ldots,F_K$.
    \item[Step 4.] Compute $D_i(\boldsymbol v)$ and
          $\operatorname{Prev}(i,\boldsymbol v)$ for $i=0,\ldots,m$ and
          $\boldsymbol v\in\mathcal Z$ by
          \eqref{eq:affine-recurrence}.
    \item[Step 5.] If $D_m(\boldsymbol N)=+\infty$, return
          $\mathrm{FAIL}$.
          Otherwise, starting from $\boldsymbol v=\boldsymbol N$, for
          $i=m,m-1,\ldots,1$, let
          $\boldsymbol u=\operatorname{Prev}(i,\boldsymbol v)$, set
          \[
              \widehat B_i=E_i\cup\bigcup_{t=1}^{K}
              \{\ell_{t,r}:u_t<r\le v_t\},
          \]
          and replace $\boldsymbol v$ by $\boldsymbol u$.
          Return $\widehat\sigma=(\widehat B_1,\ldots,\widehat B_m)$,
          $\widehat C=\max_{i\in[m]}P(\widehat B_i)$ computed with the
          original processing times, and
          $\widehat Z=D_m(\boldsymbol N)+Z_0$.
\end{description}

If $\mathcal J_+=\varnothing$, then $\boldsymbol N=\boldsymbol0$, and
the algorithm returns $(E_1,\ldots,E_m)$.

\begin{theorem}
    \label{thm:affine-qptas}
    For every $0<\varepsilon\le1$, algorithm $\mathcal A_1'$ is a
    $(1+\varepsilon,1)$-approximation algorithm for
    $P|C_{\max}\le T|\operatorname{TMC}$ with costs
    $c_{ij}=s_i+a_iw_j$.
    More precisely, every schedule returned by $\mathcal A_1'$
    satisfies $C_{\max}(\widehat\sigma)\le(1+\varepsilon)T$, even if
    $\mathcal F(T)=\varnothing$; if $\mathcal F(T)\ne\varnothing$, the
    algorithm returns a schedule with $Z(\widehat\sigma)\le Z^*(T)$;
    and if it returns $\mathrm{FAIL}$, then $\mathcal F(T)=\varnothing$.
    The running time of $\mathcal A_1'$ is $O\bigl(mK(n+1)^{2K}\bigr)$,
    where $K=\lceil\log(2n/\varepsilon)/\log(1+\varepsilon/2)\rceil$.
\end{theorem}

\begin{proof}
    If Step~1 returns $\mathrm{FAIL}$, then $\mathcal F(T)=\varnothing$
    by Section~\ref{sec:preliminaries}.
    Suppose that Step~5 returns a schedule $\widehat\sigma$.
    As in Lemma~\ref{lem:mctcm-dp-optimality}, the recovery passes only
    through states with finite values and ends at
    $(0,\boldsymbol0)$, so every job of $\mathcal J_+$ is assigned
    exactly once; every job of $\mathcal J_0$ is assigned through the
    sets $E_i$.
    By \eqref{eq:affine-transition-cost},
    $Z(\widehat\sigma)=D_m(\boldsymbol N)+Z_0$.
    The jobs of $\mathcal J_+$ on each machine have rounded load at most
    $H$, so $p_j\le\bar p_j$ and \eqref{eq:affine-tiny-load} give
    \[
        P(\widehat B_i)\le H+\delta T=(1+2\delta)T=(1+\varepsilon)T,
        \qquad i\in[m].
    \]
    This bound does not depend on whether $\mathcal F(T)$ is empty.

    Suppose that $\mathcal F(T)\ne\varnothing$, and let
    $\sigma^T=(B_1^T,\ldots,B_m^T)$ be a schedule attaining $Z^*(T)$.
    By \eqref{eq:mctcm-rounding-bound}, the jobs of $\mathcal J_+$ on
    each machine of $\sigma^T$ have rounded load at most
    $(1+\delta)T=H$.
    Lemma~\ref{lem:affine-ordered-assignment}, applied to
    $L_i=B_i^T\cap\mathcal J_+$, yields an assignment of
    $\mathcal J_+$ that is represented by the state
    $(m,\boldsymbol N)$ and whose cost does not exceed the cost of the
    jobs of $\mathcal J_+$ in $\sigma^T$.
    Hence $D_m(\boldsymbol N)$ is finite and at most this cost, and
    Step~5 returns a schedule.
    Since each job of $\mathcal J_0$ is assigned at its minimum cost,
    $Z_0$ does not exceed the cost of the jobs of $\mathcal J_0$ in
    $\sigma^T$.
    Adding the two bounds gives
    \begin{equation}
        Z(\widehat\sigma)=D_m(\boldsymbol N)+Z_0
        \le Z(\sigma^T)=Z^*(T).
        \label{eq:affine-cost-guarantee}
    \end{equation}
    In particular, if Step~5 returns $\mathrm{FAIL}$, then
    $\mathcal F(T)=\varnothing$.

    For the running time, let $V=|\mathcal Z|=\prod_{t=1}^{K}(N_t+1)
    \le(n+1)^K$.
    Step~2 takes $O(mn)$ time.
    Step~3 takes $O(n\log(n+1)+nK+KV)$ time, including the computation
    of $\bar P(\boldsymbol v)$ for all $\boldsymbol v\in\mathcal Z$.
    Step~4 considers at most $V^2$ pairs $(\boldsymbol u,\boldsymbol v)$
    for each machine, and each pair is processed in $O(K)$ time with the
    prefix sums $F_t$, so this step takes $O(mKV^2)$ time.
    Step~5 takes $O(mK+n)$ time.
    Since $K\ge1$, all terms are bounded by $O(mK(n+1)^{2K})$.
    Finally, $\delta=\varepsilon/2$ gives $n/\delta=2n/\varepsilon$,
    which yields the stated value of $K$.
\end{proof}

\subsection{The budget-constrained and lexicographic problems}
\label{subsec:affine-other}

For $P|\operatorname{TMC}\le\widehat U|C_{\max}$, the minimum total
assignment cost is attained by assigning every job to a machine of
minimum cost.
Hence a schedule satisfying the budget exists if and only if
\begin{equation}
    \sum_{j\in\mathcal J}\min_{i\in[m]}(s_i+a_iw_j)\le\widehat U.
    \label{eq:affine-budget-feasibility}
\end{equation}
The left-hand side can be computed in $O(mn)$ time, and the
corresponding schedule has makespan at most $P(\mathcal J)$, so that
$L\le C^*(\widehat U)\le P(\mathcal J)$ whenever
\eqref{eq:affine-budget-feasibility} holds.
Let $\mathcal A_2'$ be obtained from $\mathcal A_2$ by replacing the
test in Step~1 with \eqref{eq:affine-budget-feasibility} and every
call $\mathcal A_1(I,T,\eta)$ with $\mathcal A_1'(I,T,\eta)$, and let
$\mathcal A_3'$ be obtained from $\mathcal A_3$ by the latter
replacement.

\begin{theorem}
    \label{thm:affine-outer}
    Let $c_{ij}=s_i+a_iw_j$ and $0<\varepsilon\le1$.
    Algorithm $\mathcal A_2'$ is a $(1,1+\varepsilon)$-approximation
    algorithm for $P|\operatorname{TMC}\le\widehat U|C_{\max}$, and
    algorithm $\mathcal A_3'$ is a $(1+\varepsilon,1)$-approximation
    algorithm for $\operatorname{LEX}(C_{\max},\operatorname{TMC})$.
    Both run in time
    \begin{equation}
        O\Bigl(\bigl[1+\varepsilon^{-1}\log(\min\{m,n\})\bigr]
        mK_\eta(n+1)^{2K_\eta}\Bigr),
        \qquad
        K_\eta=\left\lceil\frac{\log(8n/\varepsilon)}
        {\log(1+\varepsilon/8)}\right\rceil,
        \label{eq:affine-outer-time}
    \end{equation}
    which is quasi-polynomial for each fixed $\varepsilon$.
    In particular, the algorithms $\mathcal A_2'$ with
    $0<\varepsilon\le1$ form a QPTAS for
    $P|\operatorname{TMC}\le\widehat U|C_{\max}$ with these costs.
\end{theorem}

\begin{proof}
    The proofs of Theorems~\ref{thm:problem2-guarantee}
    and~\ref{thm:problem3-guarantee} use only three ingredients: the
    characterization of budget feasibility together with
    $L\le C^*(\widehat U)\le P(\mathcal J)$; the relations
    \eqref{eq:budget-deadline}, \eqref{eq:prelim-monotone}, and
    \eqref{eq:prelim-makespan-bounds}, which hold for every cost
    function; and the properties of $\mathcal A_1$ stated in
    Theorem~\ref{thm:problem1-guarantee}.
    The first ingredient is provided by
    \eqref{eq:affine-budget-feasibility}, and
    Theorem~\ref{thm:affine-qptas} provides the same properties for
    $\mathcal A_1'$.
    The proofs therefore carry over and show that $\mathcal A_2'$
    returns a schedule with $Z(\widehat\sigma)\le\widehat U$ and
    $C_{\max}(\widehat\sigma)\le(1+\eta)^2C^*(\widehat U)$, and that
    $\mathcal A_3'$ returns a schedule with $Z(\widehat\sigma)\le Z^*$
    and $C_{\max}(\widehat\sigma)\le(1+\eta)^3C^*$.
    Since $\eta=\varepsilon/4$, we have
    $(1+\eta)^2\le(1+\eta)^3\le1+61\varepsilon/64\le1+\varepsilon$.

    By \eqref{eq:tccmm-number-of-calls}, each algorithm makes at most
    $r_{\max}+2=O(1+\varepsilon^{-1}\log(\min\{m,n\}))$ calls of
    $\mathcal A_1'$ with accuracy $\eta=\varepsilon/4$.
    In each call, $\delta=\eta/2=\varepsilon/8$ and
    $n/\delta=8n/\varepsilon$, so Theorem~\ref{thm:affine-qptas}
    bounds its running time by $O(mK_\eta(n+1)^{2K_\eta})$.
    The feasibility test and the generation of the deadlines take
    $O(mn+r_{\max})$ time, which is dominated by the calls.
    This proves \eqref{eq:affine-outer-time}.
    For fixed $\varepsilon$, $K_\eta=O(\log(n+1))$, so
    $(n+1)^{2K_\eta}=(n+1)^{O(\log(n+1))}$, and the running time is
    bounded by $|I|^{O(\log|I|)}$.
    Since $\mathcal A_2'$ never violates the budget, the algorithms
    $\mathcal A_2'$ form a QPTAS.
\end{proof}

For $c_{ij}=s_i+l_ip_j$, which is obtained with $a_i=l_i$ and
$w_j=p_j$, Theorem~\ref{thm:affine-outer} yields a QPTAS for
$P|\operatorname{TMC}\le\widehat U|C_{\max}$ that satisfies the
original budget.
Whether a PTAS exists for this cost function remains open.

\section{Conclusion}
\label{sec:conclusion}

This paper has studied non-preemptive scheduling on identical parallel
machines with assignment costs $c_{ij}=a_iw_j+b_j$, where the
processing times $p_j$ and the cost weights $w_j$ are specified
independently and the number of machines is part of the input.
For each fixed $0<\varepsilon\le1$, we obtained polynomial-time
$(1+\varepsilon,1)$-, $(1,1+\varepsilon)$-, and
$(1+\varepsilon,1)$-approximation algorithms for
$P|C_{\max}\le T|\operatorname{TMC}$,
$P|\operatorname{TMC}\le\widehat U|C_{\max}$, and
$\operatorname{LEX}(C_{\max},\operatorname{TMC})$, respectively.
In all three cases, the guarantee on the total assignment cost is
exact.

These results have the following consequences for the special cases
discussed in Section~\ref{subsec:related}.
The algorithm for $P|\operatorname{TMC}\le\widehat U|C_{\max}$ is a
PTAS that satisfies the original budget, and it applies in particular
to the extension $c_{ij}=l_ic_j$ proposed by \citet{Li2024}.
For $\operatorname{LEX}(C_{\max},\operatorname{TMC})$ with
multiplicative costs, our algorithm improves the
$(r+1,1)$-approximation of \citet{Lee2014} to a
$(1+\varepsilon,1)$-approximation; for additive costs, it improves the
$(3r/2,1)$-approximation of \citet{Kononov2019} in the same way.
Here, $r$ is the approximation ratio of the algorithm used for
$P\Vert C_{\max}$ in those studies, whose algorithms have substantially
lower running times than ours.
For costs $c_{ij}=s_i+a_iw_j$, with no ordering relation imposed
between $s_i$ and $a_i$, we obtained the same guarantees in
quasi-polynomial time by assigning every job with processing time at
most $\varepsilon T/(2n)$ to a machine of minimum cost and the remaining
jobs by dynamic programming.
This yields a QPTAS that satisfies the original budget for
$c_{ij}=s_i+l_ip_j$, the second extension proposed by \citet{Li2024}.

%\section*{Acknowledgment}

\paragraph{AI use.} The author used ChatGPT and  Claude to assist with text editing. The author is responsible for all content of the paper.

\bibliographystyle{abbrvnat}
\bibliography{references}

\end{document}